\documentclass[onefignum,onetabnum]{siamart171218}

\usepackage{amsmath}
\usepackage{enumerate}
\usepackage{mathrsfs}
\usepackage{color}
\usepackage{tikz}
\usepackage{mathtools}
\usepackage{breqn}
\usepackage{colortbl}

\def\R{{\mathbb{R}}}

\newcommand{\defword}[1]{\textcolor{blue}{\em #1}}
\newcommand{\pd}[2]{\ensuremath{ \frac{ \partial #1 }{\partial #2}  }}

\newtheorem{example}[theorem]{Example}

\newtheorem{question}[theorem]{Question}
\newtheorem{assumption}[theorem]{Assumption}

\usepackage{lipsum}
\usepackage{amsfonts}
\usepackage{graphicx}
\usepackage{epstopdf}
\usepackage{algorithmic}
\ifpdf
  \DeclareGraphicsExtensions{.eps,.pdf,.png,.jpg}
\else
  \DeclareGraphicsExtensions{.eps}
\fi

\newsiamremark{remark}{Remark}
\newsiamremark{hypothesis}{Hypothesis}
\crefname{hypothesis}{Hypothesis}{Hypotheses}
\newsiamthm{claim}{Claim}

\headers{Multistability of Small Reaction Networks}{Xiaoxian Tang, and Hao Xu}

\title{Multistability of Small  Reaction Networks\thanks{Submitted to the editors DATE.
\funding{XT was funded by the NSFC12001029.}
}}

\author{Xiaoxian Tang\thanks{School of Mathematical Sciences, Beihang University, Beijing,  China
  (\email{xiaoxian@buaa.edu.cn}, \url{https://sites.google.com/site/rootclassification/}).}
\and Hao Xu\thanks{School of Mathematics and Statistics, The University of Melbourne, Parkville, VIC 3010, Australia (\email{hao.xu3@student.unimelb.edu.au})}
}

\usepackage{amsopn}

\ifpdf
\hypersetup{
  pdftitle={Multistability of Small Reaction Networks},
  pdfauthor={Xiaoxian Tang, and Hao Xu}
}
\fi

\begin{document}

\maketitle

\begin{abstract}
 For three typical sets of small reaction networks (networks with two reactions, one irreversible and one reversible reaction, or two reversible-reaction pairs),   we completely answer the challenging question:
what is the smallest subset of all multistable networks such that any multistable network outside of the subset
contains either more species or more reactants than any network in this subset?
\end{abstract}



\section{Introduction}\label{sec:intro}
For the dynamical systems that arise from biochemical reaction networks,  we ask the following question:

\begin{question}\label{question}
Given a class of networks with the same number of irreversible reactions and the same number of reversible-reaction pairs,
what is the smallest nonemp-ty subset of multistable networks such that any multistable network outside of the subset has either more species or more reactants than any network from the subset? Here, we define the number of reactants as the maximum sum of stoichiometric coefficients in the reactant complexes (see Definition \ref{def:km}).
\end{question}
The above question is motivated by the multistationarity (multistability) problem of biochemical reaction systems, which
is crucial for understanding basic phenomena such as decision-making process in cellular signaling \cite{BF2001, FM1998, XF2003, CTF2006}.
Given a network, we pursue
rate constants such that the corresponding dynamical system arising under mass-action kinetics
has at least two (stable) positive steady states in the same stoichiometric compatibility class. Mathematically,
one needs to identify a value or an open region in the parameter space for which a parametric semi-algebraic system has at least two real solutions, which is a fundamental  problem in computational
algebraic geometry \cite{signs, CFMW}.
It is well-known that networks with only one reaction admit no multistationarity/multistability. So, for Question \ref{question}, the first non-trivial case to study is the class of networks with two reactions (possibly reversible). It is implied by \cite{Joshi:Shiu:Multistationary} that for the networks with two pairs of reversible reactions,
the smallest nonempty subclass of multistable networks is the set with a single network
``$0 ~ \xLeftrightarrow[]{} ~   X_1,
 \; 2X_1~ \xLeftrightarrow[]{} ~ 3X_1$". That means for any other network with two reversible-reaction pairs, if it admits multistability, either it has at least $2$ species, or
 the number of reactants is at least $4$.
 In this paper, our main contributions are  complete answers to Question \ref{question} for the networks with exactly two reactions (see Theorem \ref{thm:4-reactant}) and those with
 one irreversible reaction and one reversible reaction (see Theorem \ref{thm:g1}).

Our main focus is the multistability problem. Generally, multisability is a much more difficult problem than multistationarity because the standard algebraic tool for studying stability (Routh-Hurwitz criterion \cite{HTX2015}, or alternatively Li\'enard-Chipart criterion \cite{datta1978}) is computationally challenging (e.g., \cite{OSTT2019}, \cite{TF2020}).
Fortunately, for the networks with one-dimensional stoichiometric subspaces, we can determine stability by checking the trace of the Jacobian matrix (see Lemma \ref{lm:stable}).
Using the simpler criterion and the elimination theory (from algebraic geometry), we prove an upper bound for the maximum number of stable positive steady states in terms of the maximum number of positive steady states (Theorems \ref{thm:multistable}). For the networks such that the  nondegeneracy conjecture \cite[Conjecture 2.3]{Joshi:Shiu:Multistationary} is true, a lower bound for the the maximum number of stable positive steady states can be similarly obtained (Theorem \ref{thm:nondegmultistable}).
 These results show that a multistable network admits at least three positive steady states. So, the number of reactants for a multistable network should be at least three (in fact, for two-reaction networks, the number of reactants should be at least four, see Theorem \ref{thm:3-reactant}).
 A recent study on at-most-bimolecular networks \cite{OSS2020} supports our result.
These results extend \cite[Theorem 3.6 2(c)]{Joshi:Shiu:Multistationary}, which is for one-species networks, to two-reaction networks and to two-species networks with one irreversible and one reversible reaction, or with two pairs of reversible reactions. We remark that these results are based on a sign condition (see Theorem \ref{thm:sign}), which also provides one way to determine multistationarity (by checking if the determinant of the Jacobian matrix changes sign)  for small networks with
one-dimensional stoichiometric subspaces (see Corollary \ref{cry:sign}). 	There have been a long list of such criterion (without or with a steady-state parametrization), see \cite{CF2005, ShinarFeinberg2012, WiufFeliu_powerlaw, BP, signs,  CFMW, DMST2019}. One criterion in the list based on degree theory \cite[Theorem 3.12]{DMST2019}  requires the networks to admit no boundary steady states, which cannot be directly applied to two-reaction networks since if a two-reaction network admits multistationarity, then it must admit boundary steady states (see Theorem \ref{thm:nomss}).

This work can be viewed as one step toward an ambitious goal: a complete classification of  multistable networks with one-dimensional stoichiometric subspaces.
As the first step toward the big goal,  Joshi and Shiu \cite{Joshi:Shiu:Multistationary} solved the
multistationarity problem for small networks with only one species or up to two reactions (possibly reversible).
Later, Shiu and de Wolff \cite{shiu-dewolff} extended these results to nondegenerate multistationarity for two-species networks with two reactions (possibly reversible).
The idea of studying small networks is inspired by the fact that
multistability or nondegenerate multistationarity can be lifted from small networks to related large networks \cite{JS13, BP16}.
Here, our contribution is straightforward: one can directly read multistable networks with few species and few reactants from the two main results Theorem \ref{thm:4-reactant} and Theorem \ref{thm:g1}.
For two-reaction networks with up to four reactants and up to three species, there are in fact only two
kinds  of networks (but infinitely many) that are multistable. For instance, by Theorem \ref{thm:4-reactant}, we directly see the network ``$X_1 ~ \xrightarrow[]{} ~   X_2+X_3,
 \; 2X_1+X_2+X_3~ \xrightarrow[]{} ~ 3X_1$"
 admits no multistability (see more details on this example in Remark \ref{rmk:p3101}). And, for the networks with
one irreversible and one reversible reaction, if there are up to three reactants and up to two species, then
only four kinds of networks are multistable.



The rest of this paper is organized as follows.
In Section~\ref{sec:back}, we introduce
mass-action kinetics systems  arising  from reaction networks. We formally state our problem and present the main results in Section \ref{sec:ps}.
In Section~\ref{sec:dim1}, for the small networks with one-dimensional stoichiometric subspaces, we provide a sign condition (see Theorem \ref{thm:sign}), which reveals the relationship between the maximum number of positive steady states and  the maximum number of stable positive stable steady states (see Theorems \ref{thm:multistable} and \ref{thm:nondegmultistable}).
In Section~\ref{sec:g0}, we study networks with exactly two reactions. We prove a list of necessary conditions for a two-reaction network to admit multistability (for instance, see Theorems \ref{thm:nomss} and \ref{thm:3-reactant}). Based on these results, for the set of all two-reaction networks, we find the smallest
subset of all multistable networks such that any multistable network
contains either more species or more reactants than any network in this subset (see the proof of Theorem \ref{thm:4-reactant}).
We extend these results for networks with reversible reactions in Section~\ref{sec:ext} (see the proof of Theorem \ref{thm:g1}).
The proof of Theorem \ref{thm:g2} is presented in  Section \ref{sec:nig2}.
Finally, we end this paper with open problems inspired by Theorem \ref{thm:nondegmultistable}, see Section \ref{sec:dis}.

\section{Background}\label{sec:back}
\subsection{Chemical reaction networks}\label{sec:pre}



In this section, we briefly recall the standard notions and definitions on reaction networks, see \cite{CFMW, Joshi:Shiu:Multistationary} for more details.
A \defword{reaction network} $G$  (or \defword{network} for short) consists of a set of $s$ species $\{X_1, X_2, \ldots, X_s\}$ and a set of $m$ reactions:
\begin{align}\label{eq:network}
\alpha_{1j}X_1 +
 \dots +
\alpha_{sj}X_s
~ \xrightarrow{} ~
\beta_{1j}X_1 +
 \dots +
\beta_{sj}X_s,
 \;
    {\rm for}~
	j=1,2, \ldots, m,
\end{align}
where all $\alpha_{ij}$ and $\beta_{ij}$ are non-negative integers, and
$(\alpha_{1j},\ldots,\alpha_{sj})\neq (\beta_{1j},\ldots,\beta_{sj})$. We call the $s\times m$ matrix with
$(i, j)$-entry equal to $\beta_{ij}-\alpha_{ij}$ the
\defword{stoichiometric matrix} of
$G$,
denoted by ${\mathcal N}$.
We call the image of ${\mathcal N}$
the \defword{stoichiometric subspace}, denoted by $S$. Note that $S$ is a real subspace.

We denote by $x_1, x_2, \ldots, x_s$ the concentrations of the species $X_1,X_2, \ldots, X_s$, respectively.
Under the assumption of {\em mass-action kinetics}, we describe how these concentrations change  in time by following system of ODEs:
\begin{align}\label{eq:sys}
\dot{x}~=~f(x)~:=~{\mathcal N}\cdot \begin{pmatrix}
\kappa_1 \, x_1^{\alpha_{11}}
		x_2^{\alpha_{21}}
		\cdots x_s^{\alpha_{s1}} \\
\kappa_2 \, x_1^{\alpha_{12}}
		x_2^{\alpha_{22}}
		\cdots x_s^{\alpha_{s2}} \\
		\vdots \\
\kappa_m \, x_1^{\alpha_{1m}}
		x_2^{\alpha_{2m}}
		\cdots x_s^{\alpha_{sm}} \\
\end{pmatrix}~,
\end{align}
where $x$ denotes the vector $(x_1, x_2, \ldots, x_s)$,
and each $\kappa_j \in \mathbb R_{>0}$ is called a \defword{rate constant}.
 By considering the rate constants as an unknown vector $\kappa=(\kappa_1, \kappa_2, \dots, \kappa_m)$, we have polynomials $f_{i} \in \mathbb Q[\kappa, x]$, for $i=1,2, \dots, s$.

 A \defword{conservation-law matrix} of $G$, denoted by $W$, is any row-reduced $d\times s$-matrix whose rows form a basis of $S^{\perp}$, where
 $d:=s-{\rm rank}({\mathcal N})$ (note here, ${\rm rank} (W)=d$). The system~\eqref{eq:sys} satisfies $W \dot x =0$,  and both the positive orthant $\mathbb R_{>0}^s$ and its closure $\mathbb R_{\ge 0}$ are forward-invariant for the dynamics. Thus,
a trajectory $x(t)$ beginning at a non-negative vector $x(0)=x^0 \in
\mathbb{R}^s_{> 0}$ remains, for all positive time,
 in the following \defword{stoichiometric compatibility class} with respect to the  \defword{total-constant vector} $c:= W x^0 \in {\mathbb R}^d$:
\begin{align}\label{eq:pc}
{\mathcal P}_c~:=~ \{x\in {\mathbb R}_{\geq 0}^s \mid Wx=c\}.~
\end{align}
That means ${\mathcal P}_c$ is forward-invariant with
respect to the dynamics~\eqref{eq:sys}.

In this work, we mainly focus on the three families of small networks defined as
\begin{align*}
{\mathcal G}_0 &\;:=\; \{\text{the networks with exactly two reactions}, \; \text{i.e.},\; m=2 \;\text{in}\; \text{\eqref{eq:network}}\},\\
{\mathcal G}_1 &\;:=\; \{\text{the networks with one irreversible  and one reversible reaction}\},\; \text{and}\\
{\mathcal G}_2 &\;:=\; \{\text{the networks with two reversible-reaction pairs}\}.
\end{align*}
We  denote the union $\cup_{i=0}^2{\mathcal G}_i$ simply by ${\mathcal G}$.
Also, we simplify/clarify our notation \eqref{eq:network} for reversible reactions. Any $G\in {\mathcal G}_1$ has the form
\begin{align}\label{eq:networkg1}
\Sigma_{i=1}^s\alpha_{i1}X_i
~ \xLeftrightarrow[]{} ~
\Sigma_{i=1}^s\beta_{i1}X_i,
 \;\;
 \Sigma_{i=1}^s\alpha_{i2}X_i
~ \xrightarrow{} ~
\Sigma_{i=1}^s\beta_{i2}X_i,
 \end{align}
 and for any
  $G\in {\mathcal G}_2$, we denote it by
\begin{align}\label{eq:networkg2}
\Sigma_{i=1}^s\alpha_{i1}X_i
~ \xLeftrightarrow[]{} ~
\Sigma_{i=1}^s\beta_{i1}X_i,
\;\;
\Sigma_{i=1}^s\alpha_{i2}X_i
~ \xLeftrightarrow[]{} ~
\Sigma_{i=1}^s\beta_{i2}X_i.
 \end{align}

 \begin{definition}\label{def:form}
For two networks $G$ and $\hat{G}$ in ${\mathcal G}$,
we say the network $\hat{G}$ \defword{has the form of the network} $G$,
if we can obtain $\hat{G}$ from $G$ by relabeling the species.
\end{definition}
\begin{example}\label{ex:form}

For instance, the network \eqref{eq:exnet1} has the form of the network \eqref{eq:exnet2}: we obtain \eqref{eq:exnet1} from \eqref{eq:exnet2} by relabeling the two species $X_2$ and $X_3$.

\begin{align}\label{eq:exnet1}
X_1 + X_3~ \xrightarrow{} ~ X_2, \;\;\;\;  X_2 + 2X_3 ~ \xrightarrow{} ~ X_1 + 3 X_3.
\end{align}
\begin{align}\label{eq:exnet2}
 X_1 + X_2~ \xrightarrow{} ~ X_3,  \;\;\;\; 2X_2 + X_3 ~ \xrightarrow{} ~ X_1 + 3 X_2.
  \end{align}
  \end{example}

\subsection{Steady states}\label{sec:ss}
A \defword{steady state} 
of~\eqref{eq:sys} is a concentration vector
$x^* \in \mathbb{R}_{\geq 0}^s$ at which $f(x)$ on the
right-hand side of the
ODEs~\eqref{eq:sys}  vanishes, i.e., $f(x^*) =0$.
If a steady state $x^*$ has all strictly positive coordinates (i.e., $x^*\in \mathbb{R}_{> 0}^s$), then we call $x^*$ a \defword{positive steady state}.
If a steady state $x^*$ has zero coordinate(s) (i.e., $x^*\in \mathbb{R}_{\geq 0}^s\backslash \mathbb{R}_{> 0}^s$), then we call $x^*$ a \defword{boundary steady state}.
We say a steady state $x^*$ is \defword{nondegenerate} if
${\rm im}\left({\rm Jac}_f (x^*)|_{S}\right)=S$,
where ${\rm Jac}_f(x^*)$ denotes the Jacobian matrix of $f$, with respect to $x$, at $x^*$.
A steady state $x^*$ is \defword{exponentially stable} (or, simply \defword{stable} in this paper)
if it is nondegenerate, and  all non-zero eigenvalues of ${\rm Jac}_f(x^*)$ have negative real parts.




Suppose $N\in {\mathbb Z}_{\geq 0}$. A  network  \defword{admits $N$ (nondegenerate) positive steady states}  if for some rate-constant
vector $\kappa$ and for some total-constant vector $c$,  it has $N$ (nondegenerate) positive steady states  in the same stoichiometric compatibility class ${\mathcal P}_c$.
A  network  \defword{admits $N$ stable positive steady states}  if for some rate-constant
vector $\kappa$ and for some total-constant vector $c$,  it has $N$ stable positive steady states in the same stoichiometric compatibility class ${\mathcal P}_c$.

The \defword{maximum number of positive steady states} of a network $G$ is
{\footnotesize
\[cap_{pos}(G)\;:=\;\max\{N\in {\mathbb Z}_{\geq 0}\cup \{+\infty\}|G \;\text{admits}\; N\; \text{positive steady states}\}.\]
}
The \defword{maximum number of nondegenerate positive steady states} of a network $G$ is
{\footnotesize
\[cap_{nondeg}(G)\;:=\;\max\{N\in {\mathbb Z}_{\geq 0}\cup \{+\infty\}|G \;\text{admits}\; N\; \text{nondegenerate positive steady states}\}.\]
}
The \defword{maximum number of stable positive steady states} of a network $G$ is
{\footnotesize
\[cap_{stab}(G)\;:=\;\max\{N\in {\mathbb Z}_{\geq 0}\cup \{+\infty\}|G \;\text{admits}\; N\; \text{stable positive steady states}\}.\]
}
It is obvious that if $\hat{G}$ has the form of $G$,
then $cap_{pos}(\hat{G})=cap_{pos}(G)$, $cap_{nondeg}(\hat{G})=cap_{nondeg}(G)$, and
 $cap_{stab}(\hat{G})=cap_{stab}(G)$.

We say a network admits \defword{multistationarity} if $cap_{pos}(G)\geq 2$.
Notice that a network that admits more than one boundary steady states is not said to admit multistationarity in this paper.
We say a network admits \defword{nondegenerate multistationarity} if $cap_{nondeg}(G)\geq 2$. We say a network admits \defword{multistability} if $cap_{stab}(G)\geq 2$.


\subsection{Problem statement and main results}\label{sec:ps}
\begin{definition}\label{def:km}
For a non-negative integer $K$, a network $G$ with reactions defined in
\eqref{eq:network} is \defword{at-most-$K$-reactant} if
for all $j\in \{1, \ldots, m\}$, we have $\sum_{k=1}^s\alpha_{kj}\leq K$, and we say
 $G$ is \defword{$K$-reactant} (or, the \defword{number of reactants} of $G$ is $K$) if $G$ is at-most-$K$-reactant  and there exists $j\in \{1, \ldots, m\}$ such that $\sum_{k=1}^s\alpha_{kj}=K$.
\end{definition}

For the classes of networks ${\mathcal G}_0$, ${\mathcal G}_1$ and ${\mathcal G}_2$,
we provide complete answers to Question \ref{question}: see Theorem \ref{thm:4-reactant}, Theorem \ref{thm:g1} and Theorem \ref{thm:g2}.
\begin{theorem}\label{thm:4-reactant}
Given $G\in {\mathcal G}_0$, 
if $G$ has up to $3$ species and $G$ is at-most-$4$-reactant, then $G$ admits multistability if and only if $G$ has the form of one of the
two networks \eqref{eq:net1} and \eqref{eq:net2} below
\begin{align}\label{eq:net1}
X_1 + 3X_2 ~ \xrightarrow{} ~  4 X_2+X_3 ,
 \;\;\;\;  X_2 + X_3~ \xrightarrow{} ~ X_1;
\end{align}
\begin{align}\label{eq:net2}
X_1 + 2X_2 + X_3 ~ \xrightarrow{} ~ \beta_{21}X_2, \;\;\;\;
3X_3~ \xrightarrow{} ~\beta_{12}X_1+\beta_{22}X_2+\beta_{32}X_3, \; 
\end{align}
where $\beta_{21}\in \{0,1\}$, $\beta_{12}\in {\mathbb Z}_{>0}$, $\beta_{22}=\beta_{12}(2-\beta_{21})$ and $\beta_{32}=\beta_{12}+3$.
\end{theorem}

\begin{remark}\label{rmk:p5191}
We remark that for two networks with the same stoichiometric matrix but different defining ODEs, they might have different
dynamical behaviors. For instance,
consider the network
\begin{align}\label{eq:net3}
X_1 + X_2 ~ \xrightarrow{} ~  2 X_2+X_3 ,
 \;\;\;\;  X_2 + X_3~ \xrightarrow{} ~ X_1.
\end{align}
The networks \eqref{eq:net1} and \eqref{eq:net3}  have the same stoichiometric matrix.
However, Theorem \ref{thm:4-reactant} shows that the network \eqref{eq:net1} admits multistability but the network \eqref{eq:net3} does not.
\end{remark}

Theorem \ref{thm:4-reactant} means for ${\mathcal G}_0$, the answer to Question \ref{question} is
$${\mathcal H} \;:=\;\{G\;\text{has the form of the network \eqref{eq:net1} or the network \eqref{eq:net2}}\}.$$
In fact, Theorem \ref{thm:4-reactant} implies that
${\mathcal H}$ is the smallest multistable subset of ${\mathcal G}_0$ such that any other multistable network has more species or more reactions than any $G$ in ${\mathcal H}$.
 Similarly, one can understand why
Theorems \ref{thm:g1} and \ref{thm:g2} below answer Question \ref{question} for ${\mathcal G}_1$ and ${\mathcal G}_2$, respectively.

\begin{theorem}\label{thm:g1}
Given $G\in {\mathcal G}_1$, 
 if $G$ has up to $2$ species and
$G$ is at-most-$3$-reactant, then $G$ admits multistability if and only if $G$ has the form of one of the
networks listed in Rows \eqref{row:1r1i10}--\eqref{row:1r1i2}  of Table \ref{tab:net2}.
\end{theorem}

\begin{theorem}\label{thm:g2}
For $G\in {\mathcal G}_2$, 
if $G$ has only one species and
$G$ is at-most-$3$-reactant, then $G$ admits multistability if and only if $G$ has the form of the
network
\begin{align}\label{eq:unique}
0 ~ \xLeftrightarrow[]{} ~   X_1,
 \;\;\;\;  2X_1~ \xLeftrightarrow[]{} ~ 3X_1.
\end{align}
\end{theorem}

It is straightforward to prove Theorem \ref{thm:g2} by Theorem \ref{thm:multistable} (see Section \ref{sec:relation}) and \cite[Theorem 3.6]{Joshi:Shiu:Multistationary}. We provide the details in Section \ref{sec:nig2}. Here, our main contributions are Theorem \ref{thm:4-reactant} and Theorem \ref{thm:g1}. See the proofs in Section \ref{sec:stable} and Section \ref{sec:ext}.

Note that for each set ${\mathcal G}_i$, an ambitious goal is to find the subset of all multistable network. Our work provides one way to achieve the goal by detecting multistable networks in  ${\mathcal G}_0$ and ${\mathcal G}_1$ when the number of species
and the number of reactants are restricted. Unfortunately, using the approach presented in this paper, we
cannot characterize multistability for the networks in ${\mathcal G}_2$ with more than one species.

\section{Small networks with one-dimensional stoichiometric subspaces}\label{sec:dim1}
\subsection{Stability}\label{sec:consis}


\begin{assumption}\label{assumption}
For any  $G\in {\mathcal G}$ with reactions defined in \eqref{eq:network},
by the definition of reaction network,
we know $(\alpha_{11},\ldots, \alpha_{s1})\neq (\beta_{11},\ldots, \beta_{s1})$.
 Without loss of generality, we will assume 
  $\beta_{11}-\alpha_{11}\neq 0$ throughout this paper.
\end{assumption}
\begin{lemma}\label{lm:dim1}
For $G\in {\mathcal G}$,  if $G$  admits multistationarity, 
then
  there exists
$\lambda \in {\mathbb R}\backslash \{0\}$
 such that the equality 
 \begin{align}\label{eq:scalar}
\left(
\begin{array}{c}
\beta_{12}-\alpha_{12}\\
\vdots\\
\beta_{s2}-\alpha_{s2}
\end{array}
\right)\;=\;
-\lambda\left(
\begin{array}{c}
\beta_{11}-\alpha_{11}\\
\vdots\\
\beta_{s1}-\alpha_{s1}
\end{array}
\right).
\end{align}
 holds, and additionally, if $G\in {\mathcal G}_0$, then
the scalar $\lambda$  is positive.
\end{lemma}
\begin{proof}
By \cite[Lemma 4.1]{Joshi:Shiu:Multistationary}, for any $G\in {\mathcal G}_0$, if $cap_{pos}(G)\geq 1$,
then there exists $\lambda>0$ such that \eqref{eq:scalar} holds.
By \cite[Theorem 5.8]{Joshi:Shiu:Multistationary}, for any $G\in {\mathcal G}_1$, if $cap_{pos}(G)\geq 2$,
then there exists $\lambda\neq 0$ such that \eqref{eq:scalar} holds.
By \cite[Theorem 5.12]{Joshi:Shiu:Multistationary}, for any $G\in {\mathcal G}_2$, if $cap_{pos}(G)\geq 2$,
then there exists $\lambda\neq 0$ such that \eqref{eq:scalar} holds.
\end{proof}

\begin{corollary}\label{cry:dim1}
For $G\in {\mathcal G}$,  if $G$ admits multistationarity, 
then  the stoichiometric subspace of $G$ is one-dimensional.
\end{corollary}

For any network $G$, if $G$ has one-dimensional stoichiometric subspace,
then  all $f_i$'s defined in \eqref{eq:sys} will be the same polynomial up to scaling.
For instance, for $G\in {\mathcal G}_0$,
we substitute
\eqref{eq:scalar} into $f(x)$ in \eqref{eq:sys}, and we get
\begin{align}\label{eq:f}
f_i \;=\; \kappa_1\left(\beta_{i1}-\alpha_{i1}\right)\Pi_{k=1}^s x_k^{\alpha_{k1}}-\lambda\kappa_2\left(\beta_{i1}-\alpha_{i1}\right)\Pi_{k=1}^s x_k^{\alpha_{k2}}, \;\; i=1, \ldots, s.
\end{align}
By this fact, we can derive a simple criterion (Lemma \ref{lm:stable}) for the stability.

  \begin{lemma}\label{lm:stable}
For  any $G\in {\mathcal G}$, if  the stoichiometric subspace of $G$ is one-dimension-al, then
a nondegenerate steady state $x^*$ is stable if and only if
$\sum_{i=1}^s\frac{\partial f_i}{\partial x_i}|_{x=x^*}<0$. 
\end{lemma}
\begin{proof}
Since the stoichiometric subspace of $G$ is one-dimensional, there exists $\lambda \in {\mathbb R}$ such that
the equality \eqref{eq:scalar} holds.
We substitute
\eqref{eq:scalar} into $f(x)$ in \eqref{eq:sys}, and
we have
\begin{align}\label{eq:conf}
\left(\beta_{i1}-\alpha_{i1}\right)f_1 \;=\; \left(\beta_{11}-\alpha_{11}\right) f_i, \;\;\text{for}\; i=2, \ldots, s.
\end{align}
By \eqref{eq:conf}, the matrix ${\rm Jac}_{f}(x^*)$  has rank $1$, and so, it has at most one non-zero
eigenvalue.
Note also ${\rm Jac}_{f}(x^*)$ has at least one non-zero eigenvalue since $x^*$ is nondegenerate.
Therefore,
there is only one non-zero eigenvalue, which is equal to the trace of
the Jacobian matrix ${\rm Jac}_{f}(x^*)$. So, this eigenvalue has a negative real part if and only if the trace of ${\rm Jac}_{f}(x^*)$ is negative.
\end{proof}

\subsection{Sign condition}\label{sec:sign}
For any $G\in {\mathcal G}$, suppose $f(x)=(f_1(x), \ldots, f_s(x))$ is defined as in \eqref{eq:sys}, and
suppose the stoichiometric subspace of $G$ is one-dimensional.
Define $h$, the system: 
\begin{align}\label{eq:h}
h_1 ~:=~ f_1, \;\; 
h_i ~:= ~(\beta_{i1}-\alpha_{i1})x_1 - (\beta_{11}-\alpha_{11})x_i - c_{i-1}, \;\; 2\leq i\leq s,
\end{align}
where $c_1, \ldots, c_{s-1}\in {\mathbb R}$.  Here,  the linear equation $h_i$
is derived by the ODEs \eqref{eq:sys} and the linear-dependency condition \eqref{eq:conf}.

\begin{theorem}\label{thm:sign}
Given $G\in {\mathcal G}$, suppose the stoichiometric subspace of $G$ is one-dimensional.
If for a rate-constant vector $\kappa^*$ and  a total-constant vector $c^*$, $G$ has exactly $N$ distinct positive steady states $x^{(1)}, \ldots, x^{(N)}$, where $x^{(1)}, \ldots, x^{(N)}$ are ordered  according to their first coordinates (i.e., $x_1^{(1)}<\ldots<x_1^{(N)}$), and  all positive steady states are nondegenerate,
then $|{\rm Jac}_h(x^{(i)})||{\rm Jac}_h(x^{(i+1)})|< 0$ for $i\in \{1, \ldots, N-1\}$.
\end{theorem}

The goal of this subsection is to prove Theorem \ref{thm:sign}. We first prepare some lemmas. In fact,
Theorem
\ref{thm:sign} directly follows from Lemma \ref{lm:uni}, Lemma \ref{lm:jac}, Lemma \ref{lm:root} and Lemma \ref{lm:nonde}.

\begin{lemma}\label{lm:uni}
Let $g(z):=a_nz^n+\cdots+a_1z+a_0$ be a univariate polynomial in ${\mathbb R}[z]$. If the equation
$g(z)=0$ has exactly $r$ $(r\geq 2)$ distinct real roots, say
$z_1<\cdots<z_r$, and if $g'\left(z_i\right)\neq 0$ for $i\in \{1, \ldots, r\}$, then
 we have $g'(z_i)g'(z_{i+1})<0$ for $i\in \{1, \ldots, r-1\}$.
\end{lemma}

\begin{proof}
Fix any $i\in \{1, \ldots, r-1\}$, let $ h(z) $ be the univariate polynomial such that
$g(z) =(z-z_{i})(z-z_{i+1})h(z)$. Note that
    $$g'(z) \;=\; (z-z_{i+1})h(z) + (z-z_{i})h(z) + (z-z_{i})(z-z_{i+1})h'(z).$$
  So,
  $g'(z_{i})g'(z_{i+1}) \;=\; -(z_{i+1}-z_{i})^{2}h(z_{i})h(z_{i+1})$.
 If $ g'(z_{i})g'(z_{i+1}) > 0 $, then we have $ h(z_{i})h(z_{i+1}) < 0 $. Notice that $ h(z) $ is a continuous function, so there exists $ z_{0} \in (z_{i},z_{i+1}) $ such that $ h(z_{0}) = 0 $.  We know $z_{0} \ne z_{i} $ for $ i \in \{1,\ldots,r\}$, which is a  contradiction to the hypothesis that $g(z)=0$ has exactly $r$ distinct roots. Therefore, we definitely have $ g'(z_{i})g'(z_{i+1}) < 0 $.
\end{proof}

\begin{lemma}\label{lm:jh}
The determinant of the Jacobian matrix of $h$ defined in \eqref{eq:h} with respect to $x$ (denoted by $|{\rm Jac}_h|$) is equal to
\begin{align}\label{eq:cryjh}
(\alpha_{11}-\beta_{11})^{s-1}\sum_{i=1}^s\frac{\partial f_i}{\partial x_i}.
\end{align}
which is also equal to
\begin{align}\label{eq:jh}
(\alpha_{11}-\beta_{11})^{s-2}\sum_{i=1}^s(\alpha_{i1}-\beta_{i1})\frac{\partial f_1}{\partial x_i},
\end{align}
\end{lemma}
\begin{proof}
For $i=2, \ldots, s$, let $\tilde h_i=\frac{h_i}{\alpha_{11} - \beta_{11}}$, where
$h_i$ is defined in \eqref{eq:h}. Then, the determinant of
the Jacobian matrix of $h_1, \tilde h_2, \ldots, \tilde h_s$ with respect to $x$ is equal to the determinant of the reduced Jacobian matrix (see \cite[Definition 9.8]{ConradiPantea_Multi} and \cite[Proposition 9.2]{ConradiPantea_Multi}).
By \cite[Proposition 9.1]{ConradiPantea_Multi}, we know the determinant of the reduced Jacobian matrix is the sum of the $r \times r$ (here, $r:={\tt rank}({\mathcal N})=1$) principal minors of ${\rm Jac}_f$.
So, we have $|{\rm Jac}_h|=(\alpha_{11}-\beta_{11})^{s-1}\sum_{i=1}^s\frac{\partial f_i}{\partial x_i}$.
Now \eqref{eq:jh} follows by substituting equalities \eqref{eq:conf} into this expression for $|{\rm Jac}_h|$.
\end{proof}



For the system $h$ \eqref{eq:h}, define
\begin{align}\label{eq:g}
g(x_1) &~:=~ h_1(x_1, \ldots, x_s)|_{x_2=\frac{\beta_{21}-\alpha_{21}}{\beta_{11}-\alpha_{11}}x_1-\frac{c_{1}}{\beta_{11}-\alpha_{11}}, \;\ldots,\; x_s=\frac{\beta_{s1}-\alpha_{s1}}{\beta_{11}-\alpha_{11}}x_1-\frac{c_{s-1}}
{\beta_{11}-\alpha_{11}}}.
\end{align}

\begin{lemma}\label{lm:jac}
For the system $h$ \eqref{eq:h} and the polynomial $g(x_1)$ \eqref{eq:g},
if $x^*$ is a solution to $h_1(x^*)=\ldots =h_s(x^*)=0$, then
\begin{align}\label{eq:jac}
\left(\alpha_{11}-\beta_{11}\right)^{s-1}g'(x^*_1) \;=\;  |{\rm Jac}_h(x^*)|.
\end{align}
\end{lemma}

\begin{proof}
By \eqref{eq:g}, and by long division, we have
\begin{align}\label{eq:ldivision}
g(x_1) =  h_1 +\frac{1}{\beta_{11}-\alpha_{11}}\sum_{i=2}^s q_i h_i,
\end{align}
where $q_1, \ldots, q_s$ are polynomials in ${\mathbb R}[x_1, \ldots, x_s]$.
For the both sides of \eqref{eq:ldivision}, we take the derivative with respect to $x_1$:
\begin{align}
g'(x^*_1) &\;= \; \frac{\partial h_1}{\partial x_1}(x^*) +\frac{1}{\beta_{11}-\alpha_{11}}\sum_{i=2}^s q_i(x^*) \frac{\partial h_i}{\partial x_1}(x^*) \notag\\
&\;=\;\frac{\partial h_1}{\partial x_1}(x^*) +\frac{1}{\beta_{11}-\alpha_{11}}\sum_{i=2}^s q_i(x^*)\left(\beta_{i1}-\alpha_{i1}\right).\label{eq:pd1}
\end{align}
For the both sides of \eqref{eq:ldivision}, we take the derivative with respect to $x_i$ ($2\leq i\leq s$):
\begin{align}
0& \;= \; \frac{\partial h_1}{\partial x_i}(x^*) +\frac{1}{\beta_{11}-\alpha_{11}} q_i(x^*) \frac{\partial h_i}{\partial x_i}(x^*),  \notag\\
&\;= \; \frac{\partial h_1}{\partial x_i}(x^*) - q_i(x^*), \; \;i=2, \ldots, s. \label{eq:pdk}
\end{align}
By \eqref{eq:pdk}, we have $q_i(x^*)=\frac{\partial h_1}{\partial x_i}(x^*) $. So, by \eqref{eq:pd1}, we have
\begin{align}\label{eq:pd3}
g'(x^*_1)
\;=\;\frac{\partial h_1}{\partial x_1}(x^*) +\frac{1}{\beta_{11}-\alpha_{11}}\sum_{i=2}^s \frac{\partial h_1}{\partial x_i}(x^*)\left(\beta_{i1}-\alpha_{i1}\right).
\end{align}
Note that $h_1=f_1$ (see \eqref{eq:h}). By Lemma \ref{lm:jh} and \eqref{eq:pd3}, we have \eqref{eq:jac}.
  \end{proof}

  \begin{lemma}\label{lm:interval}
  Given $G\in {\mathcal G}$, suppose the stoichiometric subspace of $G$ is one-dimensional.
If for a rate-constant vector $\kappa^*$ and a total-constant vector $c^*$, $G$ has a positive steady state $x^*$, then the first coordinate
$x_1^*$ is contained in the open interval
\begin{align}\label{eq:interval}
I \;:=\; \cap_{i=2}^s I_i \;\;\; \text{where}\;\;\;
I_i \;:= \;
\begin{cases}
(\frac{c_{i-1}}{\beta_{i1}-\alpha_{i1}}, +\infty)&\text{if}\; \frac{\beta_{i1}-\alpha_{i1}}{\beta_{11}-\alpha_{11}}>0\\
(0, +\infty)&\text{if}\; \frac{\beta_{i1}-\alpha_{i1}}{\beta_{11}-\alpha_{11}}=0\\
(0, \frac{c_{i-1}}{\beta_{i1}-\alpha_{i1}})& \text{if}\;  \frac{\beta_{i1}-\alpha_{i1}}{\beta_{11}-\alpha_{11}}<0
\end{cases}.
\end{align}
\end{lemma}
\begin{proof}
By the system $h$ \eqref{eq:h}, for any $i$ $(2\leq i\leq s)$,
$x^*_i\;=\;\frac{\beta_{i1}-\alpha_{i1}}{\beta_{11}-\alpha_{11}}x^*_1-\frac{c_{i-1}}{\beta_{11}-\alpha_{11}}.$
Note that $x^*$ is positive. So, for its $i$-th coordinate $x^*_i$, we have $x^*_i>0$.
Hence, $x^*_1$ is contained in the interval $I$ defined in \eqref{eq:interval}.
\end{proof}

\begin{lemma}\label{lm:root}
 Given $G\in {\mathcal G}$, suppose the stoichiometric subspace of $G$ is one-dimensional.
If for a rate-constant vector $\kappa^*$ and a total-constant vector $c^*$, $G$ has exactly $N$ distinct positive steady states $x^{(1)}, \ldots, x^{(N)}$, where $x^{(1)}, \ldots, x^{(N)}$ are ordered  according to their first coordinates (i.e., $x_1^{(1)}<\ldots<x_1^{(N)}$), then all $x_1^{(1)},\ldots, x_1^{(N)}$  are roots to
$g(x_1)=0$, and for any $1\leq i\leq N-1$, there is no other real root to $g(x_1)=0$ between $x_1^{(i)}$ and $x_1^{(i+1)}$.
\end{lemma}
\begin{proof}
By the system $h$ \eqref{eq:h} and \eqref{eq:g}, all $x_1^{(1)},\ldots, x_1^{(N)}$  are roots to
$g(x_1)=0$.
By Lemma \ref{lm:interval}, we have $x_1^{(1)},\ldots, x_1^{(N)}\in I$ (see \eqref{eq:interval}).
Hence, if $g(x_1)=0$ has a real solution $x^*_1$   between the two solutions $x_1^{(i)}$ and $x_1^{(i+1)}$, then
$x^*_1\in I$. For $j=2, \ldots, s$, let  $x^*_j\;=\;\frac{\beta_{j1}-\alpha_{j1}}{\beta_{11}-\alpha_{11}}x^*_1-\frac{c_{j-1}}{\beta_{11}-\alpha_{11}}$. Then
$x^*$ is also a positive steady state, and it is different from $x^{(1)}, \ldots, x^{(N)}$, which is a contradiction to the hypothesis that $G$ has exactly $N$ distinct positive steady states.
\end{proof}

\begin{lemma}\label{lm:nonde}
Given $G\in {\mathcal G}$, suppose the stoichiometric subspace of $G$ is one-dimensional. A steady state $x^*$ is nondegenerate if and only if
$|{\rm Jac}_h(x^*)|\neq 0$.
\end{lemma}
\begin{proof}
Let ${\mathcal N}_1:=(\beta_{11}-\alpha_{11}, \ldots, \beta_{s1}-\alpha_{s1})^{\top}$. Recall
Assumption \ref{assumption} (we assume $\beta_{11}-\alpha_{11}\neq 0$).
By \eqref{eq:conf}, we have $(f_1, \ldots, f_s)^{\top}=\frac{1}{\beta_{11}-\alpha_{11}}f_1{\mathcal N}_1$, and hence, ${\rm Jac}_f(x^*)=\frac{1}{\beta_{11}-\alpha_{11}}{\mathcal N}_1\nabla f_1(x^*)$, where $\nabla f_1:=(\pd{f_1}{x_1}, \ldots, \pd{f_1}{x_s})$.   Therefore, we have $${\rm Jac}_f(x^*){\mathcal N}_1=\frac{1}{\beta_{11}-\alpha_{11}}({\mathcal N}_1\nabla f_1(x^*)){\mathcal N}_1=\frac{1}{\beta_{11}-\alpha_{11}}{\mathcal N}_1(\nabla f_1(x^*){\mathcal N}_1).$$
Note that by  \eqref{eq:jh} in Lemma \ref{lm:jh}, we have $|{\rm Jac}_h(x^*)|=(\alpha_{11}-\beta_{11})^{s-2}\nabla f_1(x^*){\mathcal N}_1$.
 Note also  the stoichiometric subspace $S$ is spanned by  ${\mathcal N}_1$.
 So, $im({\rm Jac}_f(x^*)|_S)=S$ if and only if  $|{\rm Jac}_h(x^*)|\neq 0$.
\end{proof}

\begin{remark}\label{rmk:defnondeg}
In \cite[Definition 9.9]{ConradiPantea_Multi}, a nondegenerate steady state is defined as a steady state such that the determinant of the reduced Jacobian matrix is non-zero.
From \cite[Proposition 9.2]{ConradiPantea_Multi} and Lemma \ref{lm:nonde}, one can see
this definition is consistent with our definition.
\end{remark}

\begin{theorem}\cite{Joshi:Shiu:Multistationary}\label{thm:nonde}
Suppose $G\in {\mathcal G}_0$, or, suppose $G\in {\mathcal G}_1\cup {\mathcal G}_2$ and $G$ has up to $2$ species.
If $cap_{pos}(G)<\infty$, then $cap_{pos}(G)=cap_{nondeg}(G)$.
\end{theorem}

\begin{corollary}\label{cry:sign}
Suppose $G\in {\mathcal G}_0$, or, suppose $G\in {\mathcal G}_1\cup {\mathcal G}_2$ and $G$ has up to $2$ species.
Assume $cap_{pos}(G)<\infty$. For a rate-constant vector $\kappa^*$, if   $|{\rm Jac}_h(x^*)|$ has the same
sign for any steady state $x^*\in {\mathbb R}_{>0}^s$, then $G$ has at most one positive steady state in any stoichiometric compatibility class.
\end{corollary}
\begin{proof}
It directly follows from Corollary \ref{cry:dim1}, Theorem \ref{thm:sign} and Theorem \ref{thm:nonde}.
\end{proof}
\subsection{Relationship between multistationarity and multistability}\label{sec:relation}

\begin{theorem}\label{thm:multistable}
Suppose $G\in {\mathcal G}$. 
If $cap_{pos}(G)=N\geq 2$ $(N\in {\mathbb Z}_{\geq 0})$, then
$cap_{stab}(G) \leq \lceil\frac{N}{2}\rceil$.
\end{theorem}
\begin{proof}
As $cap_{pos}(G)\geq 2$, it has a one-dimensional stoichiometric subspace (by Corollary  \ref{cry:dim1}), and so,
by Lemma \ref{lm:stable},  a nondegenerate steady state $x^*$ is stable if and only if $\sum_{i=1}^s\frac{\partial f_i}{\partial x_i}(x^*)<0$. By Lemma \ref{lm:jh}, we have $$|{\rm Jac}_h(x^*)|=(\alpha_{11}-\beta_{11})^{s-1}\sum_{i=1}^s\frac{\partial f_i}{\partial x_i}(x^*).$$
So, by Theorem \ref{thm:nonde}, we conclude
\begin{align}\label{eq:multistable}
\lfloor\frac{cap_{nondeg}(G)}{2}\rfloor\leq  cap_{stab}(G) \leq\lceil\frac{cap_{nondeg}(G)}{2}\rceil \leq \lceil\frac{N}{2}\rceil.
\end{align}

\end{proof}

\begin{theorem}\label{thm:nondegmultistable}
Suppose $G\in {\mathcal G}_0$, or, suppose $G\in {\mathcal G}_1\cup {\mathcal G}_2$ and $G$ has up to $2$ species.
If $cap_{pos}(G)=N\geq 2$ $(N\in {\mathbb Z}_{\geq 0})$, then
 $\lfloor\frac{N}{2}\rfloor\leq cap_{stab}(G) \leq \lceil\frac{N}{2}\rceil$.
\end{theorem}
\begin{proof}
By Theorem \ref{thm:nonde}, for $G\in {\mathcal G}_0$, or  for $G\in {\mathcal G}_1\cup {\mathcal G}_2$ ($G$ has up to $2$ species), we have $cap_{nondeg}(G)=cap_{pos}(G)=N$. So, by \eqref{eq:multistable} in the proof of Theorem \ref{thm:multistable}, we have the conclusion.
\end{proof}

\section{Networks in ${\mathcal G}_0$}\label{sec:g0}
\subsection{Boundary steady states and multistationarity}\label{sec:mss}
The main result of this subsection is Theorem \ref{thm:nomss}, which shows
a multistationary network in ${\mathcal G}_0$ must admit boundary steady states.
This result will be used in the proof of Theorem \ref{thm:4-reactant} (see Section \ref{sec:stable}). In order prove Theorem \ref{thm:nomss}, we start with
some useful lemmas.

\begin{lemma}\label{lm:jachss}
Given $G\in {\mathcal G}_0$, suppose the stoichiometric subspace of $G$ is one-dimensional.
For the two systems $f(x)$ \eqref{eq:f} and $h(x)$  \eqref{eq:h},
if $x^*\in {\mathbb R}^s$ is a solution to $f_1(x^*)=\ldots=f_s(x^*)=0$, then
\begin{align}\label{eq:jachss}
|{\rm Jac}_h(x^*)|\;=\; \kappa_1(\alpha_{11}-\beta_{11})^{s-1}\Pi^s_{k=1}{x_k^*}^{\alpha_{k1}-1}\sum_{i=1}^s (\beta_{i1}-\alpha_{i1})(\alpha_{i1}-\alpha_{i2})\Pi_{k\neq i}{x_k^*}.
\end{align}
\end{lemma}
\begin{proof}
By \eqref{eq:f}, we have
\begin{align}\label{eq:pdf}
\frac{\partial f_i}{\partial x_i} \;:=\; \kappa_1\alpha_{i1}(\beta_{i1}-\alpha_{i1})x^{-1}_i\Pi^s_{k=1}{x_k}^{\alpha_{k1}}-\lambda\kappa_2\alpha_{i2}(\beta_{i1}-\alpha_{i1})x^{-1}_i\Pi^s_{k=1}{x_k}^{\alpha_{k2}},
\end{align}
and
\begin{align}\label{eq:sf}
 \kappa_1(\beta_{i1}-\alpha_{i1})\Pi^s_{k=1}{x_k^*}^{\alpha_{k1}} \;=\;\lambda\kappa_2(\beta_{i1}-\alpha_{i1})\Pi^s_{k=1}{x_k^*}^{\alpha_{k2}}.
\end{align}
By \eqref{eq:pdf} and \eqref{eq:sf}, we have
\begin{align}\label{eq:spdf}
\frac{\partial f_i}{\partial x_i}(x^*) \;:=\; \kappa_1(\beta_{i1}-\alpha_{i1})(\alpha_{i1}-\alpha_{i2}){x_i^*}^{-1}\Pi^s_{k=1}{x_k^*}^{\alpha_{k1}}.
\end{align}
Hence, by Lemma \ref{lm:jh} and \eqref{eq:spdf}, we have \eqref{eq:jachss}.
\end{proof}

\begin{lemma}\label{lm:nonzero}
 Given $G\in {\mathcal G}_0$, suppose the stoichiometric subspace of $G$ is one-dimensional.
If $G$ admits a nondegenerate steady state, then
 the numbers in the sequence
\begin{align}\label{eq:mss}
(\beta_{11}-\alpha_{11})(\alpha_{11}-\alpha_{12}), \ldots, (\beta_{s1}-\alpha_{s1})(\alpha_{s1}-\alpha_{s2})
\end{align}
cannot be all zeros.
\end{lemma}

\begin{proof}
The conclusion follows from  Lemma \ref{lm:nonde} and Lemma \ref{lm:jachss}.
\end{proof}

\begin{lemma}\label{lm:nomss}
Given $G\in {\mathcal G}_0$, suppose the stoichiometric subspace of $G$ is one-dimensional.
If the network $G$ admits multistationarity, then
\begin{align}\label{eq:nomss}
\exists i, j \in \{1, \ldots, s\} \; \text{s.t.}\; (\beta_{i1}-\alpha_{i1})(\alpha_{i1}-\alpha_{i2})(\beta_{j1}-\alpha_{j1})(\alpha_{j1}-\alpha_{j2})<0.
\end{align}
\end{lemma}
\begin{proof}
By Theorem \ref{thm:nonde} and Lemma \ref{lm:nonzero}, if $cap_{pos}(G)\geq 1$, then the
numbers in the sequence \eqref{eq:mss} cannot be all zeros.
Then, the conclusion follows from  Lemma \ref{lm:jachss} and Corollary \ref{cry:sign}.
\end{proof}

\begin{remark}
The condition \eqref{eq:nomss} is exactly the same with the condition (2) stated in \cite[Theorem 5.2]{Joshi:Shiu:Multistationary}.
\end{remark}

\begin{lemma}\label{lm:bss}
Given $G\in {\mathcal G}_0$, 
if $G$ has no boundary steady state in any stoichiometric compatibility class, then
for any $k$ $(1\leq k\leq s)$, we have either $\alpha_{k1}=0$ or $\alpha_{k2}=0$ (i.e., the two monomials $\Pi^s_{k=1}{x_k}^{\alpha_{k1}}$ and $\Pi^s_{k=1}{x_k}^{\alpha_{k2}}$ have no common variables).
\end{lemma}
\begin{proof}
Note that for $G\in {\mathcal G}_0$, we know $h_1$ defined in \eqref{eq:h} is
\begin{align}\label{eq:h0}
h_1\;=\;\left(\beta_{11}-\alpha_{11}\right)
\left(\kappa_1\Pi_{k=1}^s x_k^{\alpha_{k1}}-\lambda\kappa_2\Pi_{k=1}^s x_k^{\alpha_{k2}}\right).
\end{align}
Clearly, if there exists $k\in \{1, \ldots, s\}$ such that $\alpha_{k1}>0$ and $\alpha_{k2}>0$,
  then there exists at least one boundary  steady state $(x_1=0, \ldots, x_s=0)$ in the
  stoichiometric compatibility class ${\mathcal P}_c$ defined by $c=(0, \ldots, 0) \in {\mathbb R}^{s-1}$, which is a contradiction to the hypothesis that $G$ has no boundary steady state in any stoichiometric compatibility class.
\end{proof}

\begin{example}\label{ex:bss}
The converse of Lemma \ref{lm:bss} might not be true. Consider the consistent network $$ X_1 + 2X_2 \xrightarrow{\kappa_1} X_2 + X_3,\;\; X_3 \xrightarrow{\kappa_2} X_1 + X_2. $$
The two monomials $\Pi^s_{k=1}{x_k}^{\alpha_{k1}}=x_1x_2^2$ and $\Pi^s_{k=1}{x_k}^{\alpha_{k2}}=x_3$ have no common variables. The system $h$ \eqref{eq:h} is
\begin{align*}
h_1 \; =\; -(\kappa_1x_1x_2^2 - \kappa_2x_3), \; h_2\;=\; -x_1+x_2 - c_1, \; h_3 \;=\; x_1+x_3 - c_2.
\end{align*}
Let $c_1=-1$, and  let $c_2=1$.
Then, for any rate constants, there is a boundary steady state $(1,0,0)$ in ${\mathcal P}_{c=(-1, 1)}$.
\end{example}

\begin{lemma}\label{lm:nocom}
Given $G\in {\mathcal G}_0$, 
if for every $k$ $(1\leq k\leq s)$, we have either $\alpha_{k1}=0$ or $\alpha_{k2}=0$, then the network $G$
does not admit multistationarity.
\end{lemma}
\begin{proof}
Assume $G$ admits multistationarity. By Corollary \ref{cry:dim1},  the stoichiometric subspace of $G$ is one-dimensional.
For any $k$ $(1\leq k\leq s)$, we have either $\alpha_{k1}=0$ or $\alpha_{k2}=0$.
If $\alpha_{k1}=0$, then $(\beta_{k1}-\alpha_{k1})(\alpha_{k1}-\alpha_{k2})=-\beta_{k1}\alpha_{k2}\leq 0$.
If $\alpha_{k2}=0$, then by Lemma \ref{lm:dim1}, there exists $\lambda>0$ such that
$$(\beta_{k1}-\alpha_{k1})(\alpha_{k1}-\alpha_{k2})=-\frac{1}{\lambda}(\beta_{k2}-\alpha_{k2})(\alpha_{k1}-\alpha_{k2})=-\frac{1}{\lambda}\beta_{k2}\alpha_{k1}\leq 0.$$
So, the signs of the non-zero numbers in the sequence \eqref{eq:mss} are all negative. By Lemma \ref{lm:nomss},  $G$
does not admit multistationarity, which is a contradiction.
\end{proof}

\begin{theorem}\label{thm:nomss}
Given $G\in {\mathcal G}_0$, 
if G has no boundary steady state in any stoichiometric compatibility class, then the network $G$
does not admit multistationarity.
\end{theorem}
\begin{proof}
If $G$ has no boundary steady state in any stoichiometric compatibility class, by Lemma \ref{lm:bss},
for any $k$ $(1\leq k\leq s)$, we have either $\alpha_{k1}=0$ or $\alpha_{k2}=0$.
 By Lemma \ref{lm:nocom},  $G$ does not admit multistationarity.
\end{proof}

\begin{remark}\label{rmk:zigzag}
Theorem \ref{thm:nomss} can be interpreted geometrically by the zig-zags proposed in \cite{Joshi:Shiu:Multistationary}.
There cannot be a zig-zag pattern whose rectangle has two sides on the axes.
\end{remark}

\subsection{Smallest multistable networks: proof of Theorem \ref{thm:4-reactant}}\label{sec:stable}
In this subsection, we present the proof of Theorem \ref{thm:4-reactant}.
First, we prove that a multistable network in ${\mathcal G}_0$ must have at least $4$ reactants (see Theorem \ref{thm:3-reactant}).
 Second, we prove a list of
necessary conditions for the multistable networks (with $3$ species) in ${\mathcal G}_0$ (see Lemma
\ref{lm:nobss4rec} and Lemma \ref{lm:3ss}).   Third, using the results proved in the first two steps, we find
all candidates for the multistable networks with with  $4$ reactants and $3$ species (see Lemma \ref{lm:4-reactant}).
Finally, we discuss the candidates stated in Lemma \ref{lm:4-reactant} one by one, and we complete the proof.
\begin{theorem}\label{thm:3-reactant}
Given $G\in {\mathcal G}_0$, if
$G$ is at-most-$3$-reactant, then $G$ does not admit multistability.
\end{theorem}

\begin{proof}
If the two monomials in $h_1(x)$ (see \eqref{eq:h0}) have no common variables,
 then
by Lemma \ref{lm:nocom}, $G$ does not admit multistationarity. So, $G$ admits no multistability.

 Note
the total degree of $h_1(x)$ with respect to $x$ is at most $3$ since $G$ is at-most-$3$-reactant.
Thus, if the two monomials in $h_1(x)$ have common variables, then
the equations
$h_1(x)=\ldots=h_s(x)=0$ have at most $2$  common positive solutions. That means
$cap_{pos}(G)\leq 2$. 
So, by Theorem \ref{thm:multistable}, $cap_{stab}(G)\leq 1$.
\end{proof}


 \begin{lemma}\label{lm:nobss4rec}
 Given $G\in {\mathcal G}_0$, if G has exactly $3$ species and
 if $cap_{stab}(G) \geq 2$, then 
 we have
\begin{align}
\beta_{k1}-\alpha_{k1} \neq 0,\; &\;\;\text{and}  \label{eq:maxtildea1}\\
\alpha_{k1}-\alpha_{k2} \neq 0, \;  &\;\;\text{for any} \; k\in \{1, 2, 3\}, \label{eq:maxtildea}
\end{align}
and we also have
\begin{align}\label{eq:fact3}
\frac{\beta_{12}-\alpha_{12}}{\beta_{11}-\alpha_{11}}\;
=\;\frac{\beta_{22}-\alpha_{22}}{\beta_{21}-\alpha_{21}}
\;=\;\frac{\beta_{32}-\alpha_{32}}{\beta_{31}-\alpha_{31}}  \;<\;0.
\end{align}
\end{lemma}
 \begin{proof}
  Clearly, if $cap_{stab}(G) \geq 2$, then $cap_{nondeg}(G)\geq 2$ and $cap_{pos}(G)\geq 2$. So, by Corollary \ref{cry:dim1}, the stoichiometric subspace of $G$ is one-dimensional, and hence
  the steady states are common solutions to the equations  $h_1(x)=h_2(x)=h_3(x)=0$ (see the system $h$ \eqref{eq:h}).
  Note that by Lemma \ref{lm:nomss},  there exist two distinct numbers $j_1, j_2\in \{1, 2, 3\}$ such that for any $k\in \{j_1, j_2\}$,
  $(\beta_{k1}-\alpha_{k1})(\alpha_{k1}-\alpha_{k2}) \neq 0$. Without loss of generality,
  assume $j_1=1$ and $j_2=2$.  Below, we show  $(\beta_{31}-\alpha_{31})(\alpha_{31}-\alpha_{32}) \neq 0$.

  In fact, we  can rewrite the equations $h_1(x)=h_2(x)=h_3(x)=0$ as
  \begin{align}
    x_2 =  \delta x_1^{-\xi_{12}} x_3^{-\xi_{32}} &:= \ell_1(x_1, x_3), \label{eq:h1x2}  \\
  x_2 =  A_2 x_1 - B_2 &:= \ell_2(x_1),  \label{eq:h2x2} \\
  x_3 =  A_3 x_1 - B_3 &:= \ell_3(x_1), \label{eq:h3x3}
  \end{align}
where $\delta = (\frac{\lambda\kappa_2}{\kappa_1})^{\frac{1}{\alpha_{21}-\alpha_{22}}} > 0,\; \xi_{i2} = \frac{\alpha_{i2}-\alpha_{i1}}{\alpha_{22}-\alpha_{21}}, (i=1,3)$, $ A_j = \frac{\beta_{j1}-\alpha_{j1}}{\beta_{11}-\alpha_{11}}$, and $B_j = \frac{c_{j-1}}{\beta_{11}-\alpha_{11}}, (j = 2,3)$.

  If $ \beta_{31}-\alpha_{31} = 0$, then the equation \eqref{eq:h3x3} becomes $x_3 = -B_3$, and so, the bivariate function $\ell_1(x_1, x_3)$ in \eqref{eq:h1x2} becomes $ \hat \ell_1(x_1) = \delta (-B_3)^{-\xi_{32}} x_1^{-\xi_{12}}$.
  Note that $\hat \ell_1(x_1)$ is a power function, and $\ell_2(x_1)$ is a linear function. There are at most 2 intersection points of their graphs in the first quadrant. So the equations $h_1(x)=h_2(x)=h_3(x)=0$ have at most $2$ common positive solutions.
   By Theorem \ref{thm:multistable}, $G$ admits no multistability, which is a contradiction. If $ \alpha_{31}-\alpha_{32} = 0 $, then $\ell(x_1, x_3)$ becomes  a power function $\tilde \ell_1(x_1)= \delta x_1^{-\xi_{12}}$. With a similar argument, we can deduce a contradiction.

   Finally, by Lemma \ref{lm:dim1}, we have \eqref{eq:fact3} since $\beta_{k1}-\alpha_{k1}\neq 0$ for all $k\in \{1, 2, 3\}$.
 \end{proof}

For $k=1, \ldots, s$, define
$\gamma_k\;:=\;\min \{\alpha_{k1}, \alpha_{k2}\}, \;\text{and} \;
\tilde \alpha_{kj}\;:=\;\alpha_{kj}-\gamma_k \; (j=1, 2)$.
Then
$h_1$ \eqref{eq:h0}  can be written as $h_1\;=\; \Pi_{k=1}^sx_k^{\gamma_k}{\tilde h}_1$, where
\begin{align}\label{eq:tildeh}
 {\tilde h}_1\;:=\;\left(\beta_{11}-\alpha_{11}\right)
\left(\kappa_1\Pi_{k=1}^s x_k^{\tilde \alpha_{k1}}-\lambda\kappa_2\Pi_{k=1}^s x_k^{\tilde \alpha_{k2}}\right).
\end{align}
Let
\begin{align}\label{eq:tildeg}
\tilde g(x_1) &~:=~ \tilde h_1(x_1, \ldots, x_s)|_{x_2=\frac{\beta_{21}-\alpha_{21}}{\beta_{11}-\alpha_{11}}x_1-\frac{c_{1}}{\beta_{11}-\alpha_{11}}, \;\ldots,\; x_s=\frac{\beta_{s1}-\alpha_{s1}}{\beta_{11}-\alpha_{11}}x_1-\frac{c_{s-1}}{\beta_{11}-\alpha_{11}}}.
\end{align}
\begin{lemma}\label{lm:jactildeg}
Given $G\in {\mathcal G}_0$, suppose the stoichiometric subspace of $G$ is one-dimensional.
Let $g(x_1)$  and $\tilde g(x_1)$ be the polynomials respectively defined in \eqref{eq:g} and \eqref{eq:tildeg}.
For a  fixed rate-constant vector $\kappa^*\in {\mathbb R}_{>0}^2$ and a total-constant vector $c^*\in {\mathbb R}^{s-1}$,
if $G$ has a positive steady state $x^*$, then $\tilde g'(x^*_1)$ has the same sign with
\begin{align}\label{eq:jactildeg}
\sum_{i=1}^s (\beta_{i1}-\alpha_{i1})(\alpha_{i1}-\alpha_{i2})\Pi_{k\neq i}{x_k^*},
\end{align}
and additionally, if $x^*$ is a stable positive steady state, then $\tilde g'(x^*_1)<0$.
\end{lemma}
\begin{proof}
Since the stoichiometric subspace of $G$ is one-dimensional,
  the steady state $x^*$ is a common solution to the equations  $h_1(x)=\ldots =h_s(x)=0$ (see \eqref{eq:h}).
By \eqref{eq:tildeg},
we have $\tilde g(x^*_1)=0$.
So, comparing \eqref{eq:g} and \eqref{eq:tildeg}, we have
$g'(x^*_1)=\Pi_{k=1}^s{x_k^*}^{\gamma_k}\tilde g'(x^*_1)$.
By Lemma \ref{lm:jac} and Lemma \ref{lm:jachss}, $g'(x^*_1)$ has the same sign with
\eqref{eq:jactildeg}, and so, $\tilde g'(x^*_1)$ has the same sign with
\eqref{eq:jactildeg}. Additionally, if $x^*$ is stable, by Lemma \ref{lm:stable}, Lemma \ref{lm:jh} and Lemma \ref{lm:jachss}, we know the sign of \eqref{eq:jactildeg} is negative, and hence, $\tilde g'(x^*_1)<0$.
\end{proof}
 \begin{lemma}\label{lm:3ss}
 Given $G\in {\mathcal G}_0$, suppose the stoichiometric subspace of $G$ is one-dimensional.
 Let $I := (a, A)$ be the interval defined in \eqref{eq:interval}, where
 $a\in {\mathbb R}$, and $A\in {\mathbb R}\cup \{+\infty\}$.
 Let  $\tilde g(x_1)$ be the polynomial defined in \eqref{eq:tildeg}.
 For a rate-constant vector $\kappa^*$ and a total-constant vector $c^*$, if  the degree of $\tilde g(x_1)$ with respect to $x_1$ is $3$, and if    
 $G$ has at least two stable positive steady states, then $\tilde{g}(x_1)$ satisfies the three conditions below:
 \begin{itemize}
 \item[(i)] $\tilde{g}(a)>0$,
 \item[(ii)]  $\tilde{g}(A)<0$ (here, $\tilde{g}(+\infty):=\lim \limits_{x_1\rightarrow +\infty} \tilde{g}(x_1)$), and
 \item[(iii)]  there exists $x_1^*\in (a, A)$ such that $\tilde{g}(x_1^*)=0$ and $\tilde{g}'(x_1^*)>0$.
 \end{itemize}
 \end{lemma}
 \begin{proof}
 By \eqref{eq:tildeg}, we clearly see that if $x^*$ is a positive steady state, then $x_1^*\in (a, A)$ and
  $\tilde g(x_1^*)=0$. If $x^*$ is stable,   then by 
  Lemma
  \ref{lm:jactildeg}, we have $\tilde g'(x_1^*)<0$. So, if $G$ has at least two stable positive steady states $x^{(1)}$ and $x^{(2)}$ (here, we assume $x_1^{(1)}<x_1^{(2)}$), then for $i\in \{1,2\}$, $x_1^{(i)}\in (a, A)$, $\tilde{g}(x_1^{(i)})=0$ and $\tilde{g}'(x_1^{(i)})<0$.
 Since  the degree of $\tilde g(x_1)$ with respect to $x_1$ is $3$, there exists
a third simple real root $x_1^*$ to the equation $\tilde g(x_1)=0$.
 By Lemma \ref{lm:uni}, we know $x_1^*\in  (x_1^{(1)}, x_1^{(2)})\subset (a, A)$, and $\tilde{g}'(x_1^*)>0$ (i.e., the statement (iii) is proved).  We can write $\tilde{g}(x_1)$ as
$$\tilde{g}(x_1) \;=\; C(x_1-x_1^{(1)})(x_1-x_1^{*})(x_1-x_1^{(2)}), \;\;\;\text{where}\; C\in {\mathbb R}.$$
So $\tilde{g}'(x_1^{(1)})<0$ implies that
$\tilde{g}'(x_1^{(1)}) \;=\; C(x_1^{(1)}-x_1^{*})(x_1^{(1)}-x_1^{(2)})<0\; (i.e., C<0). $
Thus, $\tilde g(a)=C(a-x_1^{(1)})(a-x_1^{*})(a-x_1^{(2)})>0$ (the statement (i)). Similarly, it is directly straightforward to see $\tilde g(A)<0$ (the statement(ii)).
 \end{proof}

\begin{definition}\label{def:equivalent0}
Given matrices of reactant coefficients
$\alpha =\left(
\alpha_{kj}
\right)_{s\times 2}\; \text{and}\;
\hat \alpha =\left(
\hat  \alpha_{kj}
\right)_{s\times 2}$,
which are
associated with two networks $G$ and $\hat G$ in ${\mathcal G}_0$,  we say $\alpha$ is \defword{equivalent} to $\hat \alpha$,
if there
exist finitely many matrices $\alpha^{(0)}, \ldots, \alpha^{(n)}$ such that
$\alpha^{(0)}=\alpha$, $\alpha^{(n)}=\hat{\alpha}$, and for any $i\in \{0, \ldots, n-1\}$,
we can obtain $\alpha^{(i+1)}$ from $\alpha^{(i)}$ by switching two rows or two columns of $\alpha^{(i)}$.
\end{definition}
Clearly, if a network $\hat {G}\in {\mathcal G}_0$ has the form of a network $G\in {\mathcal G}_0$,  then the two matrices of reactant coefficients associated with $G$ and $\hat {G}$ are equivalent (remark that the converse might not be true). Recall Example \ref{ex:form}. The two sets of reactant coefficients (say $\alpha$ and $\hat \alpha$)
of networks \eqref{eq:exnet1} and \eqref{eq:exnet2} can be written as matrices
\begin{align*}
\alpha \;=\;\left(
\begin{array}{cc}
1&0\\
0&1\\
1&2
\end{array}
\right)\; \text{and}\;
\hat \alpha \;=\;\left(
\begin{array}{cc}
0&1\\
2&1\\
1&0
\end{array}
\right).
\end{align*}
We can obtain $\alpha$ from $\hat \alpha$ by first switching
the two columns and then switching the last two rows.
\begin{lemma}\label{lm:4-reactant}
If a $3$-species network $G\in {\mathcal G}_0$ is at-most-$4$-reactant,
and if $G$ admits multistability, then $G$ can only have the form of one of the networks listed in Table \ref{tab:net}.
\end{lemma}

\begin{proof}
If $G$ admits multistability, then by Theorem \ref{thm:3-reactant}, $G$ must be $4$-reactant, and so, the degree of 
$$h_1\;=\;\left(\beta_{11}-\alpha_{11}\right)
\left(\kappa_1\Pi_{k=1}^3 x_k^{ \alpha_{k1}}-\lambda\kappa_2\Pi_{k=1}^3 x_k^{\alpha_{k2}}\right)\; (\text{recall} \;\eqref{eq:h0})$$
 with respect to $x$ is exactly $4$, i.e.,
\begin{align} \label{eq:4-reactant1}
\max \{\sum_{k=1}^3 \alpha_{k1}, \sum_{k=1}^3 \alpha_{k2}\} \;=\; 4.
\end{align}
By
Lemma \ref{lm:nocom}, the two monomials in $h_1$ have common variables. Recall that $\tilde h_1$  \eqref{eq:tildeh} is the polynomial such that
$h_1=\Pi_{k=1}^3x_k^{\gamma_k}{\tilde h}_1$, where $\gamma_k=\min \{\alpha_{k1}, \alpha_{k2}\}$.
So,
the degree of $\tilde h_1$ \eqref{eq:tildeh} with respect to $x$ is at most $3$. On the other hand,
Theorem \ref{thm:multistable} implies that if $G$ admits multistability, then
$cap_{pos}(G)\geq 3$. Note that all positive steady states of $G$ are common solutions to the equations $\tilde h_1(x)=h_2(x)=\ldots=h_s(x)=0$. So, the degree of $\tilde h_1$ with respect to $x$ is at least $3$.
Overall, the degree of $\tilde h_1$ with respect to $x$ is exactly $3$.
So, by the definition of $\tilde h_1$,  we have
\begin{align}
\sum_{k=1}^3 \min \{\alpha_{k1}, \alpha_{k2}\} \;=\; 1. \label{eq:4-reactant2}
\end{align}
Therefore, by Lemma \ref{lm:nomss} and Lemma \ref{lm:nobss4rec}, we know that the matrix of reactant coefficients $\alpha := (\alpha_{kj})_{3\times 2}$ and the matrix of product coefficients $\beta := (\beta_{kj})_{3\times 2}$ associated with $G$ belong to the set
{\footnotesize
\begin{align}\label{eq:mathcalC}
{\mathcal B} \;:=\; \{(\alpha, \beta)\in {\mathbb Z}^{3\times 2}_{\geq 0}\times  {\mathbb Z}^{3\times 2}_{\geq 0}\;\text{s.t.  \eqref{eq:nomss}, \eqref{eq:maxtildea}, \eqref{eq:fact3}, \eqref{eq:4-reactant1} and \eqref{eq:4-reactant2} hold}\}.
\end{align}
}
Here, we recall that the condition \eqref{eq:nomss} stated in Lemma \ref{lm:nomss} is
\begin{align*}
\exists i, j \in \{1, \ldots, 3\} \; \text{s.t.}\; (\beta_{i1}-\alpha_{i1})(\alpha_{i1}-\alpha_{i2})(\beta_{j1}-\alpha_{j1})(\alpha_{j1}-\alpha_{j2})<0,
\end{align*}
and
the conditions \eqref{eq:maxtildea}--\eqref{eq:fact3} stated  in Lemma \ref{lm:nobss4rec} are
$$\alpha_{k1}-\alpha_{k2} \neq 0 \; \text{for any} \; k\in \{1, 2, 3\},$$
and
$$\frac{\beta_{12}-\alpha_{12}}{\beta_{11}-\alpha_{11}}\;
=\;\frac{\beta_{22}-\alpha_{22}}{\beta_{21}-\alpha_{21}}
\;=\;\frac{\beta_{32}-\alpha_{32}}{\beta_{31}-\alpha_{31}}  \;<\;0.$$
Define a map $\pi: {\mathbb Z}^{3\times 2}_{\geq 0}\times  {\mathbb Z}^{3\times 2}_{\geq 0}\rightarrow {\mathbb Z}^{3\times 2}_{\geq 0}$  such that
for any $(\alpha, \beta)\in {\mathbb Z}^{3\times 2}_{\geq 0}\times  {\mathbb Z}^{3\times 2}_{\geq 0}$, $\pi(\alpha, \beta) = \alpha$.
Let
{\footnotesize
\begin{align}\label{eq:Ca}
{\mathcal B}_{\alpha} \;:=\; \{\alpha\in  {\mathbb Z}^{3\times 2}_{\geq 0}\;\text{s.t.  \eqref{eq:maxtildea}, \eqref{eq:4-reactant1} and \eqref{eq:4-reactant2} hold}\}.
\end{align}
}
Notice that ${\mathcal B}_{\alpha}$ is a finite set.
For each $\alpha\in {\mathcal B}_{\alpha}$, define its equivalence class in ${\mathcal B}_{\alpha}$ as $[\alpha] := \{\hat \alpha\in {\mathcal B}_{\alpha}| \hat \alpha  \; \text{and}\; \alpha \;\text{are equivalent matrices}\}$.
We explicitly compute the set ${\mathcal B}_{\alpha}$ by {\tt Maple2020} \cite{maple}, and
it is straightforward to check by a computer program that there are $12$ equivalence classes in ${\mathcal B}_{\alpha}$ (see
the supporting file \#7 in Table \ref {tab:sup}).
We pick a representative from each equivalence class, and we present them in Table \ref{tab:a}.

Note
\begin{align}
{\mathcal B}  \;=\;  \pi^{-1}({\mathcal B}_{\alpha})\cap {\mathcal B}
                    \;=\;  \pi^{-1}(\cup_{\alpha \in {\mathcal B}_{\alpha}} [\alpha])\cap {\mathcal B}
                     \;=\;  \cup_{\alpha \in {\mathcal B}_{\alpha}}\cup_{\hat \alpha\in [\alpha]}
                     \left(\pi^{-1}(\hat \alpha)\cap {\mathcal B} \right). \label{eq:c}
\end{align}
By Definition \ref{def:equivalent0}, if $\hat \alpha \in [\alpha]$, then
there exist two permutation matrices $P$ and $Q$ such that
$\hat \alpha=P\alpha Q$.  Thus, there exists a bijection
$\phi : \pi^{-1}(\alpha)\cap {\mathcal B} \rightarrow \pi^{-1}(\hat \alpha) \cap {\mathcal B}$
such that
for any $(\alpha, \beta)\in \pi^{-1}(\alpha)\cap {\mathcal B}$, $\phi(\alpha, \beta) := (\hat \alpha, P\beta Q)$.
By Definition
\ref{def:form}, the two networks associated with $(\alpha, \beta)$ and $\phi(\alpha, \beta)$ have the same form. Thus,
by \eqref{eq:c}, the multistable network $G$ has the form of a network associated with an element in  $\pi^{-1} (\alpha)\cap {\mathcal B}$ for a representative $\alpha$ in ${\mathcal B}_{\alpha}$.
In the rest of the proof, we explain how to compute $\pi^{-1} (\alpha)\cap {\mathcal B}$ for each representative recorded in Table \ref{tab:a}.

For the values of $\alpha_{kj}$ recorded in Table \ref{tab:a}-Row (1), 
the condition \eqref{eq:fact3} implies
\begin{align}
(\beta_{12}-1)/(\beta_{11}-2)<0,\label{eq:sfact21} \\
\beta_{22}/(\beta_{21}-1)<0, \label{eq:sfact22}\\
\beta_{32}/(\beta_{31}-1)<0. \label{eq:sfact23}
\end{align}
Note that $\beta_{kj}\in {\mathbb Z}_{\geq 0}$. So by \eqref{eq:sfact22} and \eqref{eq:sfact23}, we have $\beta_{21}=\beta_{31}=0$.
Also, note that the sequence \eqref{eq:mss} is now
 \begin{align}\label{eq:sfact1}
\beta_{11}-2, \beta_{21}-1, \;\text{and}\; \beta_{31}-1.
\end{align}
Since both $\beta_{21}-1$ and $\beta_{31}-1$ are negative, 
by \eqref{eq:nomss},  we have $\beta_{11}-2>0$. So, by \eqref{eq:sfact21}, we have $\beta_{12}=0$.
We substitute $\beta_{21}=\beta_{31}=\beta_{12}=0$ and the values of $\alpha_{kj}$ recorded in Table \ref{tab:a}-Row (1)
into \eqref{eq:fact3}, and we get
$$(\beta_{11}-2)\beta_{22}=1, \;\text{and}\; (\beta_{11}-2)\beta_{32}=1.$$
We solve $\beta_{kj}$ from these two equations over ${\mathbb Z}_{\geq 0}$, and we get
$\beta_{11}=3$, $\beta_{22}=1$ and $\beta_{32}=1$.
Above all, we conclude that
\begin{align*}
\pi^{-1}(\alpha)\cap {\mathcal B}=
\{\left(\left(
\begin{array}{cc}
2&1\\
1&0\\
1&0
\end{array}
\right),
\left(
\begin{array}{cc}
3&0\\
0&1\\
0&1
\end{array}
\right)\right)
\}.
\end{align*}
Similarly, from  each $\alpha$ recorded in  each row of Table \ref{tab:a}, we can solve the corresponding  $\pi^{-1}(\alpha)\cap {\mathcal B}$, and we
record the corresponding network in Table \ref{tab:net}.
\end{proof}

\newcounter{rowc}
\newcommand{\newrowc}{\refstepcounter{rowc}\arabic{rowc}}

\begin{table}[]
  \renewcommand\arraystretch{1.25}
  \tiny
  \centering
  \begin{tabular}{ccccccc}
  \hline
  \cellcolor[HTML]{FFCE93}& \cellcolor[HTML]{FFCE93}$\alpha_{11}$ &\cellcolor[HTML]{FFCE93} $\alpha_{21}$ &\cellcolor[HTML]{FFCE93} $\alpha_{31}$ & \cellcolor[HTML]{FFCE93}$\alpha_{12}$ &\cellcolor[HTML]{FFCE93} $\alpha_{22}$ &\cellcolor[HTML]{FFCE93} $\alpha_{32}$ \\ \hline
  (\newrowc\label{row:c1})   & $2$ & $1$ & $1$ & $1$ & $0$ & $0$  \\ \hline
  (\newrowc\label{row:c2.1}) & $2$ & $2$ & $0$ & $1$ & $0$ & $1$  \\ \hline
  (\newrowc\label{row:c2.2}) & $1$ & $3$ & $0$ & $0$ & $1$ & $1$  \\ \hline
  (\newrowc\label{row:c2.3}) & $1$ & $2$ & $1$ & $0$ & $0$ & $2$  \\ \hline
  (\newrowc\label{row:c3.1}) & $2$ & $2$ & $0$ & $1$ & $0$ & $2$  \\ \hline
  (\newrowc\label{row:c3.2}) & $1$ & $3$ & $0$ & $0$ & $1$ & $2$  \\ \hline
  (\newrowc\label{row:c3.3}) & $1$ & $2$ & $1$ & $0$ & $0$ & $3$  \\ \hline
  (\newrowc\label{row:c4.1}) & $2$ & $2$ & $0$ & $1$ & $0$ & $3$  \\ \hline
  (\newrowc\label{row:c4.2}) & $1$ & $3$ & $0$ & $0$ & $1$ & $3$  \\ \hline
  (\newrowc\label{row:c4.3}) & $1$ & $2$ & $1$ & $0$ & $0$ & $4$  \\ \hline
  (\newrowc\label{row:c5.1}) & $4$ & $0$ & $0$ & $1$ & $1$ & $1$  \\ \hline
  (\newrowc\label{row:c5.2}) & $3$ & $1$ & $0$ & $0$ & $2$ & $1$  \\ \hline
  \end{tabular}
  \caption{Representatives of equivalence classes in ${\mathcal B}_{\alpha}$ \eqref{eq:Ca}}
  \label{tab:a}
  \end{table}

{\bf Proof of Theorem \ref{thm:4-reactant}.}
``$\Leftarrow$": 
For the network \eqref{eq:net1}, it is straightforward to check that
the equality \eqref{eq:scalar} holds for $\lambda=1$.
 Let $\kappa_1=9$, $\kappa_2=50$, $c_1=6$, and $c_2=\frac{59}{10}$.
By solving the equations $h_1(x)=h_2(x)=h_3(x)=0$ (see \eqref{eq:h} and \eqref{eq:h0}), we see that the network has three nondegenerate positive steady states:
{\scriptsize
\[x^{(1)}=(\frac{7}{2}-\frac{\sqrt{205}}{6},
    \frac{5}{2}+\frac{\sqrt{205}}{6},
    \frac{12}{5}+\frac{\sqrt{205}}{6}), x^{(2)}=(5,1,\frac{9}{10}),
x^{(3)}=(\frac{7}{2}+\frac{\sqrt{205}}{6},
    \frac{5}{2}-\frac{\sqrt{205}}{6},
    \frac{12}{5}-\frac{\sqrt{205}}{6}). \]
    }
It is straightforward to check by Lemma \ref{lm:stable} that
$x^{(1)}$ and $x^{(3)}$ are stable.

For the network \eqref{eq:net2}, if $\beta_{21}=0$, then for any $\beta_{12}\in {\mathbb Z}_{>0}$, $\beta_{22}=2\beta_{12}$ and
$\beta_{32}=\beta_{12}+2$.
It is straightforward to check that
the equality \eqref{eq:scalar} holds for $\lambda=\beta_{12}>0$.
  Let $\kappa_1=1$, $\kappa_2=\frac{48}{\beta_{12}}$, $c_1=\frac{13}{2}$, and $c_2=\frac{1}{4}$.
Then we have
\begin{align*}
h_1&\;=\;\left(\beta_{11}-\alpha_{11}\right)
\left(\kappa_1\Pi_{k=1}^s x_k^{\alpha_{k1}}-\lambda\kappa_2\Pi_{k=1}^s x_k^{\alpha_{k2}}\right)\;=\;-\left(x_1x^2_2x_3-48x_3^3\right), \\
h_2& \;=\; (\beta_{21}-\alpha_{21})x_1 - (\beta_{11}-\alpha_{11})x_2 - c_{1}\;=\;-2x_1 +x_2 -\frac{13}{2}, \;\;\;\text{and}\\
h_3 &\;=\; (\beta_{31}-\alpha_{31})x_1 - (\beta_{11}-\alpha_{11})x_3 - c_{2}\;=\;-x_1 +x_3 - \frac{1}{4}.
\end{align*}
By solving the equations $h_1(x)=h_2(x)=h_3(x)=0$, the network has three nondegenerate positive steady states:
{\scriptsize
\[x^{(1)}=(\frac{19}{8}-\frac{3\sqrt{33}}{8},
    \frac{45}{4}-\frac{3\sqrt{33}}{4},
    \frac{21}{8}-\frac{3\sqrt{33}}{8}),
x^{(2)}=(\frac{3}{4},8,1),
x^{(3)}=(\frac{19}{8}+\frac{3\sqrt{33}}{8},
    \frac{45}{4}+\frac{3\sqrt{33}}{4},
    \frac{21}{8}+\frac{3\sqrt{33}}{8}). \]
    }
It is straightforward to check by Lemma \ref{lm:stable} that
$x^{(1)}$ and $x^{(3)}$ are stable. Similarly, if $\beta_{21}=1$, then for any $\beta_{12}\in {\mathbb Z}_{>0}$, $\beta_{22}=\beta_{12}$ and
$\beta_{32}=\beta_{12}+2$.
  Let $\kappa_1=1$, $\kappa_2=
  \frac{12}{\beta_{12}}$, $c_1=\frac{13}{4}$, and $c_2=\frac{1}{4}$.
Then the network has three nondegenerate positive steady states:
{\scriptsize
\[x^{(1)}=(\frac{19}{8}-\frac{3\sqrt{33}}{8},
    \frac{45}{8}-\frac{3\sqrt{33}}{8},
    \frac{21}{8}-\frac{3\sqrt{33}}{8}),
x^{(2)}=(\frac{3}{4},4,1),
x^{(3)}=(\frac{19}{8}+\frac{3\sqrt{33}}{8},
    \frac{45}{8}+\frac{3\sqrt{33}}{8},
    \frac{21}{8}+\frac{3\sqrt{33}}{8}). \]
}
It is straightforward to check by Lemma \ref{lm:stable} that
$x^{(1)}$ and $x^{(3)}$ are stable.
Here, we compute these steady states by {\tt Maple2020} \cite{maple}, see the supporting file \#6 in Table \ref{tab:sup}.

``$\Rightarrow$":
By Theorem \ref{thm:nondegmultistable} and \cite[Theorem 3.6 2(b), Theorem 4.8]{Joshi:Shiu:Multistationary},
if $G\in {\mathcal G}_0$ and $G$ has up to $2$ species, then $G$ admits no multistability.
The networks \eqref{eq:net1} and \eqref{eq:net2} are listed in Row (\ref{row:2r2.2}) and Row (\ref{row:2r3.3}) of Table \ref{tab:net}, respectively.
By Lemma \ref{lm:4-reactant}, we only need to show none of the other networks listed in Table \ref{tab:net}  admits multistability.

For the network in Table \ref{tab:net}-Row (\ref{row:2r1}), the polynomial $\tilde g(x_1)$ defined in \eqref{eq:tildeg} is
$$\tilde g(x_1)\; =\; \kappa_1x_1x_2x_3 - \lambda \kappa_2\;|_{x_2=-x_1-c_1, x_3=-x_1-c_2},$$
where $\lambda:=-\frac{\beta_{12}-\alpha_{12}}{\beta_{11}-\alpha_{11}}>0$, and the interval $I$ defined in \eqref{eq:interval} is $(0, \min\{-c_1, -c_2\})$. Note that $\tilde g(0) = -\lambda \kappa_2<0$ for any $\kappa_2\in {\mathbb R}_{>0}$.  So, by Lemma \ref{lm:3ss},
 this network in Row (\ref{row:2r1}) does not admit multistability.

For the network in Table \ref{tab:net}-Row (\ref{row:2r2.1}), the polynomial $\tilde g(x_1)$ is
$$\tilde g(x_1)\; =\; \kappa_1x_1x_2^2 - \lambda \kappa_2x_3\;|_{x_2=\frac{\beta_{21}-2}{\beta_{11}-2}x_1-\frac{c_1}{\beta_{11}-2}, x_3=\frac{\beta_{31}}{\beta_{11}-2}x_1-\frac{c_2}{\beta_{11}-2}},$$
where $\lambda:=-\frac{\beta_{12}-\alpha_{12}}{\beta_{11}-\alpha_{11}}>0$,
and the interval $I$ is $(\max\{0, \frac{c_2}{\beta_{31}}\}, \frac{c_1}{\beta_{21}-2})$.
From the second column of Row (\ref{row:2r2.1}), we see that
$\beta_{11}-2>0$,  $\beta_{21}-2=-\beta_{22}(\beta_{11}-2)<0$, and $\beta_{31}=\beta_{11}-2>0$. If $\frac{c_2}{\beta_{31}}<0$, then $\tilde g(0)=\lambda\kappa_2\frac{c_2}{\beta_{11}-2}<0$, and so,
by Lemma \ref{lm:3ss} (i),  the network in Row (\ref{row:2r2.1}) does not admit multistability.
If $\frac{c_2}{\beta_{31}}
\geq 0$, then by Lemma \ref{lm:jactildeg}, for any positive steady state $x^*$ of $G$, $\tilde g'(x^*_1)$ has the same sign with
\begin{align}\label{eq:4-reactant3}
&\sum_{i=1}^3 (\beta_{i1}-\alpha_{i1})(\alpha_{i1}-\alpha_{i2})\Pi_{k\neq i}{x_k^*} \notag\\
=& \; (\beta_{11}-2)x_2^*x_3^*+2(\beta_{21}-2)x_1^*x_3^*-\beta_{31}x_1^*x_2^* \notag \\
=& \;\left((\beta_{11}-2)x_3^*-\beta_{31}x_1^*\right)x_2^* + 2(\beta_{21}-2)x_1^*x_3^* \notag\\
=& \;-c_2x_2^* + 2(\beta_{21}-2)x_1^*x_3^*,  \notag
\end{align}
which is negative (Note $\beta_{21}-2<0$).
So by Lemma \ref{lm:3ss} (iii),  the network in Row (\ref{row:2r2.1}) does not admit multistability. Similarly, we can prove the networks in Rows (\ref{row:2r3.1}), (\ref{row:2r4.1}), and (\ref{row:2r5.2}) do not admit multistability.

For the network in Table \ref{tab:net}-Row (\ref{row:2r2.3}), the polynomial $\tilde g(x_1)$ is
$$\tilde g(x_1)\; =\; -\left(\kappa_1x_1x_2^2 - \lambda \kappa_2x_3\right)\;|_{x_2=-(\beta_{21}-2)x_1+c_1, x_3=x_1+c_2},$$
where $\lambda:=-\frac{\beta_{12}-\alpha_{12}}{\beta_{11}-\alpha_{11}}>0$,
and the interval $I$ is $(\max\{0, \frac{c_1}{\beta_{21}-2},-c_2\}, +\infty)$.
From the second column of Row (\ref{row:2r2.3}), we see that $\beta_{21}-2<0$.
If $-c_2>0$ and $-c_2>\frac{c_1}{\beta_{21}-2}$, then by the fact that $\tilde g(-c_2)=\kappa_1c_2\left(\left(\beta_{21}-2\right)c_2+c_1\right)^2\leq 0$ and
by Lemma \ref{lm:3ss} (i),  the network in Row (\ref{row:2r2.3}) does not admit multistability.
If $-c_2\leq 0$, then by Lemma \ref{lm:jactildeg}, for any positive steady state $x^*$ of $G$, $\tilde g'(x^*_1)$ has the same sign with
\begin{align}
&\sum_{i=1}^3 (\beta_{i1}-\alpha_{i1})(\alpha_{i1}-\alpha_{i2})\Pi_{k\neq i}{x_k^*} \notag\\
=& \; -x_2^*x_3^*+2(\beta_{21}-2)x_1^*x_3^*+x_1^*x_2^* \;\;\;\;\;\;\;\;\;\;\;\;\;\;\;\;\;\;\;\;\;\;\;\;\;\;\;\;\;\;\notag \\
=& \;\left(x_1^*-x_3^*\right)x_2^* + 2(\beta_{21}-2)x_1^*x_3^* \notag\\
=& \;-c_2x_2^* + 2(\beta_{21}-2)x_1^*x_3^*\; <\; 0 \notag
\end{align}
Similarly, if $-c_2\leq \frac{c_1}{\beta_{21}-2}$ (i.e., $c_1+c_2(\beta_{21}-2)\leq 0$), then we also have
\begin{align}
&\sum_{i=1}^3 (\beta_{i1}-\alpha_{i1})(\alpha_{i1}-\alpha_{i2})\Pi_{k\neq i}{x_k^*}\notag\\
=&-x_2^*x_3^*+2(\beta_{21}-2)x_1^*x_3^*+x_1^*x_2^* \notag \\
=& \;-x_2^*x_3^*+(\beta_{21}-2)x_1^*x_3^*+x_1^*\left((\beta_{21}-2)x_3^*+x_2^*\right) \notag\\
=& \;-x_2^*x_3^*+(\beta_{21}-2)x_1^*x_3^*+x_1^*\left(c_1+c_2(\beta_{21}-2)\right) <0 \notag
\end{align}
(note the last equality $(\beta_{21}-2)x_3^*+x_2^*=c_1+c_2(\beta_{21}-2)$ above is deduced by eliminating $x^*_1$ from the two conservation law equations $(\beta_{21}-2)x^*_1+x^*_2-c_1=0$ and $-x^*_1+x^*_3-c_2=0$).
So by Lemma \ref{lm:3ss} (iii),  the network in Row (\ref{row:2r2.3}) does not admit multistability. Similarly, we can prove
the network in Row (\ref{row:2r4.3}) does not admit multistability.

For the network in Table \ref{tab:net}-Row (\ref{row:2r3.2}), the polynomial $\tilde g(x_1)$ is
$$\tilde g(x_1)\; =\; -(\kappa_1x_1x_2^2 - \lambda \kappa_2x_3^2)\;|_{x_2=-x_1+c_1, x_3=-\beta_{31}x_1+c_2},$$
where $\lambda:=-\frac{\beta_{12}-\alpha_{12}}{\beta_{11}-\alpha_{11}}>0$,
and the interval $I$ is $(0, \min \{c_1, \frac{c_2}{\beta_{31}}\})$
(from the second column of Row (\ref{row:2r2.3}), we see that $\beta_{31}>0$).
If $c_1<\frac{c_2}{\beta_{31}}$, then by the fact that $\tilde g(c_1)=\lambda\kappa_2(-\beta_{31}c_1+c_2)^2\geq 0$ and
by Lemma \ref{lm:3ss} (ii),  the network in Row (\ref{row:2r2.1}) does not admit multistability.
If $\frac{c_2}{\beta_{31}}\leq c_1$ (i.e., $-\beta_{31}c_1+c_2\leq 0$), then by Lemma \ref{lm:jactildeg}, for any positive steady state $x^*$ of $G$, $\tilde g'(x^*_1)$ has the same sign with
\begin{align}
&\sum_{i=1}^3 (\beta_{i1}-\alpha_{i1})(\alpha_{i1}-\alpha_{i2})\Pi_{k\neq i}{x_k^*}\notag \\
=& \; -x_2^*x_3^*+2x_1^*x_3^*-2\beta_{31}x_1^*x_2^* \;\;\; \;\;\;\;\;\;\;\;\;\;\;\;\;\;\;\;\;\;\;\;\;\;\;\;\;\;\;\;\;\;\notag \\
=& \;-x_2^*x_3^*+2x_1^*\left(x_3^*-\beta_{31}x_2^*\right) \notag\\
=& \;-x_2^*x_3^*+2x_1^*\left(-\beta_{31}c_1+c_2\right)\; <\; 0 \notag
\end{align}
(note the last equality $x_3^*-\beta_{31}x_2^*=-\beta_{31}c_1+c_2$ above is deduced by eliminating $x^*_1$ from the two conservation law equations $x^*_1+x^*_2-c_1=0$ and $\beta_{31}x^*_1+x^*_3-c_2=0$).
So by Lemma \ref{lm:3ss} (iii),  the network in Row (\ref{row:2r3.2}) does not admit multistability. Similarly, the network in Row (\ref{row:2r4.2}) does not admit multistability.

For the network in Table \ref{tab:net}-Row (\ref{row:2r5.1}), the polynomial $\tilde g(x_1)$ is
$$\tilde g(x_1)\; =\; (\beta_{11}-4)\left(\kappa_1x_1^3 - \lambda \kappa_2x_2x_3\right)\;|_{x_2=\frac{\beta_{21}}{\beta_{11}-4}x_1-\frac{c_1}{\beta_{11}-4}, x_3=\frac{\beta_{31}}{\beta_{11}-4}x_1-\frac{c_2}{\beta_{11}-4}},$$
where $\lambda:=-\frac{\beta_{12}-\alpha_{12}}{\beta_{11}-\alpha_{11}}>0$,
and the interval $I$ is $(\max\{0, \frac{c_1}{\beta_{21}},\frac{c_2}{\beta_{31}}\}, +\infty)$. Note that $\tilde g(+\infty)$ has a positive sign for any $\kappa_2\in {\mathbb R}_{>0}$.  So, by Lemma \ref{lm:3ss} (ii),
the network in Row (\ref{row:2r5.1}) does not admit multistability. $\Box$
\newcounter{rowa}
\newcommand{\newrowa}{\refstepcounter{rowa}\arabic{rowa}}
\begin{table}[] 
  \renewcommand\arraystretch{1.75}
  \tiny
  \centering
  \caption{All candidates for multistable networks in ${\mathcal G}_0$ with $4$ reactants and $3$ species}
  (only networks in Row (\ref{row:2r2.2}) and Row (\ref{row:2r3.3}) are multistable)
  \label{tab:net}
  \begin{tabular}{cc|r}
  \hline
\cellcolor[HTML]{FFCE93}  &\cellcolor[HTML]{FFCE93}  Network
  &\cellcolor[HTML]{FFCE93} $ (\beta_{11}, \beta_{21}, \beta_{31}, \beta_{12}, \beta_{22}, \beta_{32})\in {\mathbb Z}^6_{\geq 0} $
	\\ \hline

  \begin{tabular}[c]{@{}c@{}}
  	(\newrowa\label{row:2r1}) \\ \quad
  \end{tabular} &
  \begin{tabular}[c]{@{}c@{}}
  	$ 2X_1 + X_2 + X_3  \rightarrow 3X_1 $   \\
    $ X_1   \rightarrow X_2 + X_3 $
  \end{tabular} &
  \begin{tabular}[r]{@{}r@{}}
  \end{tabular} \\  \hline

  \begin{tabular}[c]{@{}c@{}}
  	(\newrowa\label{row:2r2.1}) \\ \quad
  \end{tabular} &
  \begin{tabular}[c]{@{}c@{}}
	$ 2X_1 + 2X_2 \rightarrow \beta_{11}X_1 + \beta_{21}X_2 + \beta_{31}X_3 $   \\
	$  X_1 +  X_3 \rightarrow \beta_{22}X_2 $
  \end{tabular} &
  \begin{tabular}[r]{@{}r@{}}
	$ (\beta_{11}, -\beta_{22}(\beta_{11}-2)+2, \beta_{11}-2, 0, \beta_{22}, 0) $ \\
	$ \beta_{11} \in\{3,4\}; \;\;\;\;0<\beta_{22}{\le}\frac{2}{\beta_{11}-2} $
  \end{tabular} \\ \hline

  \begin{tabular}[c]{@{}c@{}}
  \cellcolor[HTML]{6CDEFF}(\newrowa\label{row:2r2.2}) \\ \quad
  \end{tabular} &
  \begin{tabular}[c]{@{}c@{}}
	$  X_1 + 3X_2 \rightarrow 4X_2 + X_3$  \\
	$  X_2 +  X_3 \rightarrow  X_1 $
  \end{tabular}  &
  \begin{tabular}[r]{@{}r@{}}
  \end{tabular}     \\ \hline

  \begin{tabular}[c]{@{}c@{}}
  	(\newrowa\label{row:2r2.3}) \\ \quad
  \end{tabular} &
  \begin{tabular}[c]{@{}c@{}}
	$  X_1 + 2X_2 +  X_3 \rightarrow \beta_{21}X_2 $    \\
	$ 2X_3 \rightarrow \beta_{12}X_1 + \beta_{22}X_2 + \beta_{32}X_3 $
  \end{tabular} &
  \begin{tabular}[r]{@{}r@{}}
	$ (0, \beta_{21}, 0, \beta_{12}, \beta_{12}(2-\beta_{21}), \beta_{12}+2 ) $ \\
	$ \beta_{21} \in \{ 0,1\}; \;\;\;\;\beta_{12}>0 $
  \end{tabular} \\ \hline

  \begin{tabular}[c]{@{}c@{}}
  	(\newrowa\label{row:2r3.1}) \\ \quad
  \end{tabular} &
  \begin{tabular}[c]{@{}c@{}}
	$ 2X_1 + 2X_2 \rightarrow \beta_{11}X_1 + \beta_{21}X_2 + \beta_{31}X_3 $   \\
	$  X_1 + 2X_3 \rightarrow \beta_{22}X_2 + \beta_{32}X_3 $
  \end{tabular} &
  \begin{tabular}[r]{@{}r@{}}
	$ (\beta_{11}, -\beta_{22}(\beta_{11}-2)+2, (2-\beta_{32})(\beta_{11}-2), 0, \beta_{22}, \beta_{32} )$ \\
	$ \beta_{11} \in \{ 3,4\};\;\;\;\;0<\beta_{22}{\le}\frac{2}{\beta_{11}-2};\;\;\;\; \beta_{32}\in \{0, 1\} $
  \end{tabular} \\ \hline

  \begin{tabular}[c]{@{}c@{}}
  	(\newrowa\label{row:2r3.2}) \\ \quad
  \end{tabular} &
  \begin{tabular}[c]{@{}c@{}}
	$  X_1 + 3X_2 \rightarrow 4X_2 + \beta_{31}X_3 $   \\
	$  X_2 + 2X_3 \rightarrow  X_1 + \beta_{32}X_3 $
  \end{tabular} &
  \begin{tabular}[r]{@{}r@{}}
	$ (0, 4, \beta_{31}, 1, 0, 2-\beta_{31} )$ \\
	$ \beta_{31} \in \{1,2\} $
  \end{tabular} \\ \hline

  \begin{tabular}[c]{@{}c@{}}
  \cellcolor[HTML]{6CDEFF}(\newrowa\label{row:2r3.3}) \\ \quad
  \end{tabular} &
  \begin{tabular}[c]{@{}c@{}}
	$  X_1 + 2X_2 + X_3 \rightarrow \beta_{21}X_2 $ \\
    $ 3X_3  \rightarrow \beta_{12}X_1 + \beta_{22}X_2 + \beta_{32}X_3 $
  \end{tabular} &
  \begin{tabular}[r]{@{}r@{}}
	$(0, \beta_{21}, 0, \beta_{12}, \beta_{12}(2-\beta_{21}), \beta_{12}+3 )$ \\
	$ \beta_{21} \in  \{0,1\};\;\;\;\; \beta_{12}>0 $
  \end{tabular} \\ \hline

  \begin{tabular}[c]{@{}c@{}}
  	(\newrowa\label{row:2r4.1}) \\ \quad
  \end{tabular} &
  \begin{tabular}[c]{@{}c@{}}
	$ 2X_1 + 2X_2 \rightarrow \beta_{11}X_1 + \beta_{21}X_2 + \beta_{31}X_3 $ \\
	$  X_1 + 3X_3 \rightarrow \beta_{22}X_2 + \beta_{32}X_3 $
  \end{tabular} &
  \begin{tabular}[r]{@{}r@{}}
	$ (\beta_{11}, -\beta_{22}(\beta_{11}-2)+2, (3-\beta_{32})(\beta_{11}-2), 0, \beta_{22}, \beta_{32} )$ \\
	$ \beta_{11} \in \{3,4\};\;\;\;\;0<\beta_{22}{\le}\frac{2}{\beta_{11}-2};\;\;\;\; \beta_{32}\in \{0, 1, 2\} $
  \end{tabular} \\ \hline

  \begin{tabular}[c]{@{}c@{}}
  	(\newrowa\label{row:2r4.2}) \\ \quad
  \end{tabular} &
  \begin{tabular}[c]{@{}c@{}}
	$  X_1 + 3X_2 \rightarrow 4X_2 + \beta_{31}X_3 $   \\
	$  X_2 + 3X_3 \rightarrow  X_1 + \beta_{32}X_3 $
  \end{tabular} &
  \begin{tabular}[r]{@{}r@{}}
	$ (0, 4, \beta_{31}, 1, 0, 3-\beta_{31} )$ \\
	$ \beta_{31} \in \{1,2,3\} $
  \end{tabular} \\ \hline

  \begin{tabular}[c]{@{}c@{}}
  	(\newrowa\label{row:2r4.3}) \\ \quad
  \end{tabular} &
  \begin{tabular}[c]{@{}c@{}}
	$  X_1 + 2X_2 +  X_3  \rightarrow \beta_{21}X_2 $   \\
	$ 4X_3 \rightarrow \beta_{12}X_1 + \beta_{22}X_2 + \beta_{32}X_3 $
  \end{tabular} &
  \begin{tabular}[r]{@{}r@{}}
	$(0, \beta_{21}, 0, \beta_{12}, \beta_{12}(2-\beta_{21}), \beta_{12}+4 )$ \\
	$\beta_{21} \in \{0,1\};\;\;\;\; \beta_{12}>0 $
  \end{tabular} \\ \hline

  \begin{tabular}[c]{@{}c@{}}
  	(\newrowa\label{row:2r5.1}) \\ \quad
  \end{tabular} &
  \begin{tabular}[c]{@{}c@{}}
	$ 4X_1 \rightarrow \beta_{11}X_1 + \beta_{21}X_2 + \beta_{31}X_3 $ \\
	$  X_1 + X_2 + X_3 \rightarrow 0 $
  \end{tabular} &
  \begin{tabular}[r]{@{}r@{}}
	$ (\beta_{21}+4, \beta_{21}, \beta_{21}, 0, 0, 0 )$ \\
	$\beta_{21}>0$
  \end{tabular} \\ \hline

  \begin{tabular}[c]{@{}c@{}}
  	(\newrowa\label{row:2r5.2}) \\ \quad
  \end{tabular} &
  \begin{tabular}[c]{@{}c@{}}
  	$ 3X_1 + X_2 \rightarrow \beta_{11}X_1 +  X_3 $  \\
  	$ 2X_2 + X_3 \rightarrow \beta_{12}X_1 + 3X_2 $
  \end{tabular} &
  \begin{tabular}[r]{@{}r@{}}
	$(\beta_{11}, 0, 1, 3-\beta_{11}, 3, 0 ) $ \\
	$ \beta_{11} \in \{0,1,2\} $
  \end{tabular} \\ \hline

  \end{tabular}


\end{table}

\begin{remark}\label{rmk:p3101}
It is stated in \cite[Remark 5.4]{Joshi:Shiu:Multistationary} that the network
recorded in Table \ref{tab:net}--Row \eqref{row:c1} admits multistability. But here, in the proof of Theorem \ref{thm:4-reactant}, we proved it does not.
\end{remark}

\section{Networks in ${\mathcal G}_1$: proof of Theorem \ref{thm:g1}}\label{sec:ext}
In this section, the goal is to prove Theorem
\ref{thm:g1}. We first recall the well-known  Descartes' rule of signs (see Theorem \ref{thm:des}).
The idea of the proof is similar to the proof of Theorem \ref{thm:4-reactant}. In Lemma \ref{lm:g1}, we find
all candidates for the multistable networks with $2$ species and $3$ reactants in ${\mathcal G}_1$ by \cite[Theorem 3.5]{shiu-dewolff} (i.e., Lemma \ref{lm:JSthm58})  and a list of necessary conditions (see Lemma \ref{lm:nec1}).
Then, we discuss these candidates one by one, and the Descartes' rule of signs plays a key role in this discussion.
\begin{definition}
	The \defword{sign} of a real number $a\in\R$ is
	\begin{equation*}
		{\rm{sign}}(a) :=  \left\{
		\begin{array}{cl}
			+ \;\; & \text{if}\;\;  a > 0  \\
			0 \;\; & \text{if}\;\;  a = 0  \\
			- \;\; & \text{if}\;\;  a < 0
		\end{array}
		\right.
	\end{equation*}
	We define the \defword{sign} of a vector $x\in\R^{n}$ as: $${\rm {sign}}(x):=({\rm{sign}}(x_1),..., {\rm{sign}}(x_n))\in\{ +,0,- \}^n .$$ The \defword{number of sign changes} in such a vector of signs $v\in\{ +,0,- \}^n$ is obtained by first removing all $0$'s from $v$ and then counting the number of times in the resulting vector a coordinate switches from $+$ to $-$ or from $-$ to $+$.
\end{definition}

\begin{theorem}(Descartes' rule of signs)\label{thm:des}\cite{Descartes}
	Given a nonzero univariate real polynomial $g(z)=a_0+a_{1}z+...+a_{n}z^{n}$, the number of positive real roots of $g$, counted with multiplicity, is bounded above by the number of sign changes in the ordered sequence of the coefficients ${\rm{sign}}(a_0),...,{\rm{sign}}(a_n)$, i.e., discard the $0$'s in this sequence and then count the number of times two consecutive signs differ.
\end{theorem}

\begin{lemma}\label{lm:JSthm58}\cite[Theorem 3.5]{shiu-dewolff}
Given $G\in {\mathcal G}_1$, if $G$ has exactly $2$ species,
then $G$ admits nondegenerate  multistationarity if and only if
there exists $\lambda \in {\mathbb R}\backslash \{0\}$ such that the equality \eqref{eq:scalar} holds for $s=2$, and
\begin{align}\label{eq:js}
\exists k\in \{1, 2\}\; \text{s.t.}\;
 \max \{\alpha_{k1}, \beta_{k1}\} < \alpha_{k2} <\beta_{k2}\; or
 \;
 \min \{\alpha_{k1}, \beta_{k1}\} > \alpha_{k2} > \beta_{k2}.
\end{align}
\end{lemma}



\begin{lemma}\label{lm:nec1}
Suppose $G\in {\mathcal G}_1$, and suppose $G$ has exactly $2$ species.
If $G$ admits multistability, then 
we have
\begin{align}\label{eq:nec1}
(\beta_{11}-\alpha_{11})(\beta_{21}-\alpha_{21})\neq 0, \;\;\;\text{and}
\end{align}
\begin{align}\label{eq:nec12}
\frac{\beta_{12}-\alpha_{12}}{\beta_{11}-\alpha_{11}}\;=\; \frac{\beta_{22}-\alpha_{22}}{\beta_{21}-\alpha_{21}} \;\neq \; 0.
\end{align}
\end{lemma}
\begin{proof}
Recall that we have $\beta_{11}-\alpha_{11}\neq 0$ by Assumption \ref{assumption}.
If $\beta_{21}-\alpha_{21}= 0$, then by \eqref{eq:g}, we have
$g(x_1)=\left(\beta_{11}-\alpha_{11}\right)
\left(\kappa_1\Gamma^{\alpha_{21}}x_1^{\alpha_{11}}-\kappa_2\Gamma^{ \beta_{21}}x_1^{ \beta_{11}}+\lambda\kappa_3 \Gamma^{ \alpha_{12}}x_1^{\alpha_{12}}\right), $
where $\Gamma=-\frac{c_{1}}{\beta_{11}-\alpha_{11}}$.
So,  $g(x_1)$ has at most $3$ terms, and hence, the number of sign changes of the coefficients is at most $2$. By Descartes' rule of signs (Theorem \ref{thm:des}),
$g(x_1)=0$ has at most $2$ positive roots. So, the network $G$ admits at most $2$ positive steady states. By Theorem \ref{thm:multistable}, the network does not admit multistability, which is a contradiction. So, the inequality \eqref{eq:nec1} holds.
Finally, by Lemma \ref{lm:dim1}, we have \eqref{eq:nec12}.
\end{proof}

\begin{definition}\label{def:equivalent1}
Given two matrices of reactant coefficients
\begin{align*}
\sigma \;=\;\left(
\begin{array}{ccc}
\alpha_{11}&\beta_{11}&\alpha_{12}\\
\alpha_{21}&\beta_{21}&\alpha_{22}
\end{array}
\right)\; \text{and}\;
\hat \sigma \;=\;\left(
\begin{array}{ccc}
\hat \alpha_{11}&\hat \beta_{11}&\hat \alpha_{12}\\
\hat \alpha_{21}&\hat \beta_{21}&\hat \alpha_{22}
\end{array}
\right),
\end{align*}
which are
associated with two 2-species networks $G$ and $\hat G$ in ${\mathcal G}_1$, we say $\sigma$ is \defword{strongly equivalent} to $\hat \sigma$,
if there
exist finitely many matrices $\sigma^{(0)}, \ldots, \sigma^{(n)}$ such that
$\sigma^{(0)}=\sigma$, $\sigma^{(n)}=\hat{\sigma}$, and for any $i\in \{0, \ldots, n-1\}$,
we can obtain $\sigma^{(i+1)}$ from $\sigma^{(i)}$ by switching the two rows or the first two columns of $\sigma^{(i)}$.
\end{definition}
\begin{example}\label{ex:equivalent1}
Consider the two networks below.
\begin{align}\label{eq:exnet3}
X_1 + 2X_2~ \xLeftrightarrow[]{}~  0, \;\;\;\;  2X_1 \xrightarrow[]{} 3X_1 + 2X_2.
\end{align}
\begin{align}\label{eq:exnet4}
   0~ \xLeftrightarrow[]{}~ 2X_1 + X_2 , \;\;\;\; 2X_2 \xrightarrow[]{}  2X_1 +3X_2.
  \end{align}
 The two matrices of reactant coefficients 
of networks \eqref{eq:exnet3} and \eqref{eq:exnet4} can be rewritten as
$\sigma \;=\;\left(
\begin{array}{ccc}
1&0&0\\
2&0&2
\end{array}
\right)\; \text{and}\;
\hat \sigma \;=\;\left(
\begin{array}{ccc}
0&2&2\\
0&1&0
\end{array}
\right)$.
We can obtain $\sigma$ from $\hat \sigma$ by first switching
the first two columns and then switching the two rows.
So $\sigma$ is strongly equivalent to $\hat \sigma$.
\end{example}
\begin{table}[]
  \tiny
  \centering
  \caption{Representatives of equivalence classes in ${\mathcal C}_{\sigma}$ \eqref{eq:Csigma}}
  \label{tab:sigma}
  \begin{tabular}{c|c|c|c}
  \hline

\cellcolor[HTML]{FFCE93}
\begin{tabular}[r]{@{}r@{}}
 see ``reverse1.mw"
 \end{tabular}
&
\cellcolor[HTML]{FFCE93}
 \begin{tabular}[r]{@{}r@{}}
 see ``reverse2.mw"
 \end{tabular}
 &
 \cellcolor[HTML]{FFCE93}
  \begin{tabular}[r]{@{}r@{}}
 see ``reverse3.mw"
 \end{tabular}
   &
   \cellcolor[HTML]{FFCE93}
   \begin{tabular}[r]{@{}r@{}}
 see ``reverse4.mw"
 \end{tabular}
  \\ \hline

$\left(
\begin{array}{ccc}
0&1&3\\
3&0&0
\end{array}
\right)$&
   \cellcolor[HTML]{FFE7FD}$\left(
\begin{array}{ccc}
{\mathbf
0}&{\mathbf
1}&{\mathbf
3}\\
{\mathbf
0}&{\mathbf
1}&{\mathbf
0}
\end{array}
\right)$

&
$\left(
\begin{array}{ccc}
0&1&2\\
3&0&0
\end{array}
\right)$
&
  $\left(
\begin{array}{ccc}
0&1&2\\
3&2&0
\end{array}
\right)$
  \\ \hline

  $\left(
\begin{array}{ccc}
0&1&2\\
3&0&1
\end{array}
\right)$&
$  \cellcolor[HTML]{FFE7FD}\left(
\begin{array}{ccc}
{\mathbf
0}&{\mathbf
1}&{\mathbf
2}\\
{\mathbf
0}&{\mathbf
1}&{\mathbf
1}
\end{array}
\right)$
&
 $\left(
\begin{array}{ccc}
0&1&2\\
3&1&0
\end{array}
\right)$
&
 $\left(
\begin{array}{ccc}
0&1&3\\
3&2&0
\end{array}
\right)$

\\\hline

$\left(
\begin{array}{ccc}
0&1&3\\
3&1&0
\end{array}
\right)$ &
$\left(
\begin{array}{ccc}
0&1&3\\
1&0&0
\end{array}
\right)$
&
  \cellcolor[HTML]{FFE7FD}
 $\left(
\begin{array}{ccc}
{\mathbf
1}&{\mathbf
0}&{\mathbf
2}\\
{\mathbf
2}&{\mathbf
0}&{\mathbf
0}
\end{array}
\right)$
&
 $\left(
\begin{array}{ccc}
0&2&3\\
3&1&0
\end{array}
\right)$

  \\ \hline

$\left(
\begin{array}{ccc}
0&2&3\\
2&0&0
\end{array}
\right)$&
  \cellcolor[HTML]{FFE7FD}
$\left(
\begin{array}{ccc}
{\mathbf
0}&{\mathbf
1}&{\mathbf
2}\\
{\mathbf
1}&{\mathbf
0}&{\mathbf
1}
\end{array}
\right)$
&
 \cellcolor[HTML]{FFE7FD}$\left(
\begin{array}{ccc}
{\mathbf
1}&{\mathbf
0}&{\mathbf
2}\\
{\mathbf
2}&{\mathbf
1}&{\mathbf
0}
\end{array}
\right)$
&

\\\hline

$  \cellcolor[HTML]{FFE7FD}\left(
\begin{array}{ccc}
{\mathbf
2}&{\mathbf
0}&{\mathbf
3}\\
{\mathbf
1}&{\mathbf
0}&{\mathbf
0}
\end{array}
\right)$&
$\left(
\begin{array}{ccc}
0&2&3\\
1&0&0
\end{array}
\right)$
&&
\\ \hline

$  \cellcolor[HTML]{FFE7FD}\left(
\begin{array}{ccc}
{\mathbf
1}&{\mathbf
0}&{\mathbf
3}\\
{\mathbf
2}&{\mathbf
0}&{\mathbf
0}
\end{array}
\right)$&
$\left(
\begin{array}{ccc}
1&0&3\\
0&2&0
\end{array}
\right)$&&
\\ \hline

$  \cellcolor[HTML]{FFE7FD}\left(
\begin{array}{ccc}
{\mathbf
1}&{\mathbf
0}&{\mathbf
2}\\
{\mathbf
2}&{\mathbf
0}&{\mathbf
1}
\end{array}
\right)$&
$  \cellcolor[HTML]{FFE7FD}\left(
\begin{array}{ccc}
{\mathbf
1}&{\mathbf
0}&{\mathbf
2}\\
{\mathbf
0}&{\mathbf
2}&{\mathbf
1}
\end{array}
\right)$
&&
\\ \hline

$  \cellcolor[HTML]{FFE7FD}\left(
\begin{array}{ccc}
{\mathbf
1}&{\mathbf
0}&{\mathbf
3}\\
{\mathbf
2}&{\mathbf
1}&{\mathbf
0}
\end{array}
\right)$&
$\left(
\begin{array}{ccc}
1&0&3\\
1&2&0
\end{array}
\right)$
&&

\\ \hline

$\left(
\begin{array}{ccc}
2&0&3\\
1&2&0
\end{array}
\right)$&
$\left(
\begin{array}{ccc}
2&0&3\\
0&2&0
\end{array}
\right)$
&&

\\ \hline
 \end{tabular}


\end{table}

\begin{lemma}\label{lm:g1}
If a $2$-species network $G\in {\mathcal G}_1$ is at-most-$3$-reactant,
and if $G$ admits multistability, then $G$ can only have the form of one of the networks listed in Table \ref{tab:net2}.
\end{lemma}

\begin{proof}
 If $G$ admits multistability, then by Theorem \ref{thm:multistable}, we have
$cap_{pos}(G)\geq 3$. Note that all positive steady states of $G$ are common solutions to the equations $h_1(x)=\ldots=h_s(x)=0$ (see \eqref{eq:h}). So, the degree of $ h_1$ with respect to $x$ is at least $3$. Since $G$ is at most 3-reactant, the degree of $h_1$ with respect to $x$ is at most $3$. Overall, the degree $h_1$ with respect to $x$ is exactly $3$, i.e.,
\begin{align}
\max \{\sum_{k=1}^2 \alpha_{k1},  \sum_{k=1}^2 \beta_{k1}, \sum_{k=1}^2  \alpha_{k2}\} \;=\; 3. &\label{eq:3-reactant1}
\end{align}
So, we have $cap_{pos}(G)=3$. That means $G$ has no boundary steady states and so, 
\begin{align}
\min \{\alpha_{11}, \beta_{11}, \alpha_{21}\}\;=\; 0,\; \text{and} \;\min \{\alpha_{21}, \beta_{21}, \alpha_{22}\}\;=\; 0. \label{eq:3-reactant2}
\end{align}
Therefore, by Lemma \ref{lm:JSthm58} and  Lemma \ref{lm:nec1}, we know that the matrix of reactant coefficients and product coefficients $\tau :=
\left(
\begin{array}{cccc}
\alpha_{11}&\beta_{11}&\alpha_{12}&\beta_{12}\\
\alpha_{21}&\beta_{21}&\alpha_{22}&\beta_{22}
\end{array}
\right)
$
associated with $G$ belong to the set
{\footnotesize
\begin{align}
{\mathcal C} \;:=\; \{\left(
\begin{array}{cccc}
\alpha_{11}&\beta_{11}&\alpha_{12}&\beta_{12}\\
\alpha_{21}&\beta_{21}&\alpha_{22}&\beta_{22}
\end{array}
\right)
\in {\mathbb Z}^{2\times 4}_{\geq 0}\;\text{s.t.
\eqref{eq:js}, \eqref{eq:nec1}, \eqref{eq:nec12},  \eqref{eq:3-reactant1}, \eqref{eq:3-reactant2} hold}\}.
\end{align}
}Define a map $\pi_{\sigma}:{\mathbb Z}^{2\times 4}_{\geq 0}\rightarrow {\mathbb Z}^{2\times 3}_{\geq 0}$  such that
for any $\left(
\begin{array}{cccc}
\alpha_{11}&\beta_{11}&\alpha_{12}&\beta_{12}\\
\alpha_{21}&\beta_{21}&\alpha_{22}&\beta_{22}
\end{array}
\right)
\in {\mathbb Z}^{2\times 4}_{\geq 0}$,
its image under $\pi_{\sigma}$ is
$
\left(
\begin{array}{ccc}
\alpha_{11}&\beta_{11}&\alpha_{12}\\
\alpha_{21}&\beta_{21}&\alpha_{22}
\end{array}
\right).
$
Note that if $\tau$ satisfies \eqref{eq:js}, then
$\pi_{\sigma}(\tau)$ satisfies
\begin{align}
\exists
k\in \{1, 2\} \;\text{s.t.}\;
\max \{\alpha_{k1}, \beta_{k1}\}\; <\; \alpha_{k2}, \; &\text{or} \;
\min \{\alpha_{k1}, \beta_{k1}\} \;>\; \alpha_{k2}\; >\;  0. &  \label{eq:maxa2}
\end{align}
Let
{\footnotesize
\begin{align}\label{eq:Csigma}
{\mathcal C}_{\sigma} \;:=\; \{
\left(
\begin{array}{ccc}
\alpha_{11}&\beta_{11}&\alpha_{12}\\
\alpha_{21}&\beta_{21}&\alpha_{22}
\end{array}
\right)\in  {\mathbb Z}^{2\times 3}_{\geq 0}\;\text{s.t. \eqref{eq:nec1}, \eqref{eq:3-reactant1}, \eqref{eq:3-reactant2}, and \eqref{eq:maxa2} hold}\}.
\end{align}
}
For any $\sigma\in {\mathcal C}_{\sigma}$, define a equivalence class ${\mathcal C}_{\sigma}$ as $$[\sigma]_{{\mathcal C}} := \{\hat \sigma\in {\mathcal C}_{\sigma}| \hat \sigma \;\text{is strongly equivalent to}\; \sigma\}.$$ It is straightforward to check by a computer program  that there are $25$ equivalence classes in ${\mathcal C}_{\sigma}$ (see the supporting files \#2--5 in Table \ref {tab:sup}), and we pick one element from each equivalence class as a representative.
We present
the $25$ representatives in Table \ref{tab:sigma}.

In Table \ref{tab:sigma}, for any representative $\sigma$ recorded in a unbold/uncolored cell,
the set  $\pi_{\sigma}^{-1} (\sigma)\cap {\mathcal C}$ is empty.
For instance, for the first column of the second row,
we have
$$\sigma\;=\;
\left(
\begin{array}{ccc}
\alpha_{11}&\beta_{11}&\alpha_{12}\\
\alpha_{21}&\beta_{21}&\alpha_{22}
\end{array}
\right)\;=\;
\left(
\begin{array}{ccc}
0&1&2\\
3&0&1
\end{array}
\right),
$$
which satisfies the condition \eqref{eq:maxa2} because for $k=1$,
$\max \{\alpha_{k1}, \beta_{k1}\} <\alpha_{k2}$ holds.
By the condition \eqref{eq:js}, we have $\beta_{12}>\alpha_{12}=2$.
The condition \eqref{eq:nec12} can be written as
 $\frac{\beta_{12}-2}{1}=\frac{\beta_{22}-1}{-3}$.
So, we have $\beta_{22}-1<0$ since $\beta_{12}-2>0$. Hence, $\beta_{22}=0$ is the only solution for $\beta_{22}$ in
${\mathbb Z}_{\geq 0}$. So we have $\beta_{12}-2=\frac{1}{3}$, and we have no solution for $\beta_{12}$ in
${\mathbb Z}_{\geq 0}$. Therefore, $\pi_{\sigma}^{-1} (\sigma)\cap {\mathcal C}=\emptyset$. Similarly, we can easily verify
that $\pi_{\sigma}^{-1} (\sigma)\cap {\mathcal C}$ is empty for any other $\sigma$ recorded in a unbold cell (see ``reverse1.mw"--``reverse4.mw").
We repeat the representatives in the bold/colored cells in the first column of Table \ref{tab:net2}, and we write down their corresponding networks in the second column.

Note
\begin{align}\label{eq:cs}
{\mathcal C}  \;=\;  \pi_{\sigma}^{-1}({\mathcal C}_{\sigma})\cap {\mathcal C}
                     \;=\;  \cup_{\sigma \in {\mathcal C}_{\sigma}}\cup_{\hat \sigma\in [\sigma]_{\mathcal C}}
                     \left(\pi_{\sigma}^{-1}(\hat \sigma)\cap {\mathcal C} \right).
\end{align}
By Definition \ref{def:equivalent1}, if $\hat \sigma \in [\sigma]_{\mathcal C}$, then
there exist two permutation matrices $P$ and $Q$ such that
$\hat \sigma=P\sigma Q$.  Thus, there exists a bijection
$\phi \;:\; \pi_{\sigma}^{-1}(\sigma)\cap {\mathcal C} \rightarrow \pi_{\sigma}^{-1}(\hat \sigma) \cap {\mathcal C}$
such that
for any $\tau\in \pi_{\sigma}^{-1}(\sigma)\cap {\mathcal C}$, $\phi(\tau) := P\tau Q$.
By Definition
\ref{def:form}, the two networks associated with $\tau$ and $\phi(\tau)$ have the same form. Thus,
by \eqref{eq:cs}, any multistable network $G$ has the form of a network associated with an element in  $\pi_{\sigma}^{-1} (\sigma)\cap {\mathcal C}$ for a representative $\sigma$  recorded in the first column of Table \ref{tab:net2}.
In the rest of the proof, we explain how to compute $\pi_{\sigma}^{-1} (\sigma)\cap {\mathcal C}$ for each representative in ${\mathcal C}_{\sigma}$ recorded in  Table \ref{tab:net2}.

For the reactant coefficients recorded in  Table \ref{tab:net2}-Row (1),
the matrix $\sigma$ is
$$
\left(
\begin{array}{ccc}
\alpha_{11}&\beta_{11}&\alpha_{12}\\
\alpha_{21}&\beta_{21}&\alpha_{22}
\end{array}
\right)\;=\;
\left(
\begin{array}{ccc}
2&0&3\\
1&0&0
\end{array}
\right),
$$
which satisfies the condition \eqref{eq:maxa2} because for $k=1$,
$\max \{\alpha_{k1}, \beta_{k1}\} <\alpha_{k2}$ holds.
By the condition \eqref{eq:js}, we have $\beta_{12}>\alpha_{12}=3$.
The condition \eqref{eq:nec12} can be written as
$
\frac{\beta_{12}-3}{-2}\;=\; \frac{\beta_{22}}{-1}$, i.e.,
$\beta_{22}=\frac{1}{2}(\beta_{12}-3)$.
Above all, we conclude that
\begin{align*}
\pi_{\sigma}^{-1}(\sigma)\cap {\mathcal C}=
\{\left(
\begin{array}{cccc}
2&0&3&\beta_{12}\\
1&0&0&\frac{1}{2}(\beta_{12}-3)\\
\end{array}
\right)|\beta_{12} \in {\mathbb Z}_{>3}
\}.
\end{align*}
Similarly, from  each set of reactant coefficients recorded in the first column of  Table \ref{tab:net2}, we can solve  $\pi_{\sigma}^{-1}(\sigma)\cap {\mathcal C}$, and we
record the corresponding $\beta_{21}$ and $\beta_{22}$  in the third column.
\end{proof}

{\bf Proof of Theorem \ref{thm:g1}.}
``$\Leftarrow$":
For the network in Table \ref{tab:net2}--Row \eqref{row:1r1i10},
it is straightforward to check that for any $\beta_{12}>2$,
the equality \eqref{eq:scalar} holds for $\lambda=-(\beta_{12}-2)<0$.
  Let $\kappa_1=\frac{1}{2}$, $\kappa_2=16$, $\kappa_3=\frac{3}{2\left(\beta_{12}-2\right)}$ and $c_1=-9$.
Then we have
\begin{align*}
h_1&\;=\;\left(\beta_{11}-\alpha_{11}\right)
\left(\kappa_1x_2-\kappa_2x_1-\lambda \kappa_3x_1^2x_2\right)\;=\;\frac{1}{2}x_2-16x_1+\frac{3}{2}x_1^2x_2,\;\;\;\text{and} \\
h_2& \;=\; (\beta_{21}-\alpha_{21})x_1 - (\beta_{11}-\alpha_{11})x_2 - c_{1}\;=\;-x_1 -x_2 +9.
\end{align*}
By solving the equations $h_1(x)=h_2(x)=0$, the network has three nondegenerate positive steady states:
{\scriptsize
$x^{(1)}=(4-\sqrt{13},
    5+\sqrt{13}),\;\;
 x^{(2)}=(1,8),\;\;
x^{(3)}=(4+\sqrt{13},
    5-\sqrt{13})$.
    }
It is straightforward to check by Lemma \ref{lm:stable} that
$x^{(1)}$ and $x^{(3)}$ are stable. Similarly, we can show
the networks in Rows \eqref{row:1r1i11}--\eqref{row:1r1i2} admit multistability. We present the
computation in the supporting file \#1, see in Table \ref{tab:sup}.

``$\Rightarrow$":
By Theorem \ref{thm:nondegmultistable} and \cite[Theorem 3.6 2(b)]{Joshi:Shiu:Multistationary},
if $G\in {\mathcal G}_1$ and $G$ has only $1$ species, then $G$ admits no multistability. By Lemma \ref{lm:g1}, we only need to show the networks listed in  Table \ref{tab:net2}-Rows (\ref{row:1r1i4})--(\ref{row:1r1i9}) do not admit multistability.


For the network in Table \ref{tab:net2}-Row (\ref{row:1r1i4}),
the polynomial $g(x_1)$ defined in \eqref{eq:g} is
\begin{align*}
g(x_1)\; &=\; -(\kappa_1x_1^2x_2 - \kappa_2-\lambda \kappa_3x_1^3)\;|_{x_2=(x_1+c_1)/2} \\
&= -(\frac{\kappa_1}{2}-\lambda\kappa_3)x_1^3-\frac{c_1\kappa_1}{2}x_1^2 +\kappa_2,
\end{align*}
where
$\lambda:=-\frac{\beta_{12}-\alpha_{12}}{\beta_{11}-\alpha_{11}}
=-\frac{\beta_{22}-\alpha_{22}}{\beta_{21}-\alpha_{21}}>0$.
The number of sign changes of the coefficients  is at most $2$ since $g(x_1)$ has at most $3$ terms. By Descartes' rule of signs (Theorem \ref{thm:des}),
$g(x_1)=0$ has at most $2$ positive roots. So, this network admits at most $2$ positive steady states and by Theorem \ref{thm:multistable}, the network does not admit multistability.


For the network in Table \ref{tab:net2}-Row (\ref{row:1r1i5}),
\begin{align*}
g(x_1)\; &=\; -(\kappa_1x_1x_2^2 - \kappa_2-\lambda \kappa_3x_1^3)\;|_{x_2=2x_1+c_1}\\
&= (\lambda\kappa_3 - 4\kappa_1)x_1^3 - 4c_1\kappa_1x_1^2 - c_1^2
\kappa_1x_1 + \kappa_2,
\end{align*}
and the interval $I$ is $(\max\{0, -\frac{c_1}{2}\}, +\infty)$, where $\lambda:=-\frac{\beta_{12}-\alpha_{12}}{\beta_{11}-\alpha_{11}}
=-\frac{\beta_{22}-\alpha_{22}}{\beta_{21}-\alpha_{21}}>0$.
If $-\frac{c_1}{2}<0$, then $c_1>0$, and so the number of sign changes of the coefficients is at most $2$ since $g(x_1)$ has at most $3$ terms. By Descartes' rule of signs (Theorem \ref{thm:des}),
$g(x_1)=0$ has at most $2$ positive roots.
Similarly, if $-\frac{c_1}{2}>0$ and $\lambda\kappa_3 - 4\kappa_1>0$, then by Descartes' rule of signs,
$g(x_1)=0$ has at most $2$ positive roots.
If $-\frac{c_1}{2}>0$ and $\lambda\kappa_3 - 4\kappa_1<0$, then
\begin{align*}
g'(-\frac{c_1}{2})\;&=\; 3(\lambda\kappa_3 - 4\kappa_1)x_1^2 - 8c_1\kappa_1x_1 - c_1^2\kappa_1|_{x_1=-\frac{c_1}{2}}\;=\;\frac{3}{4}\lambda\kappa_3c^2_1>0.
\end{align*}
So $g'(x_1)=0$ has at most $1$ root over the interval $I$, and hence $g(x_1)=0$ has at most $2$ roots over $I$. Above all,
the network admits at most $2$ positive steady states, and so, by Theorem \ref{thm:multistable}, the network does not admit multistability.
Similarly, we can show that the network in  Table \ref{tab:net2}-Row (\ref{row:1r1i6}) does not admit multistability.

For the network in Table \ref{tab:net2}-Row (\ref{row:1r1i7}), the polynomial $g(x_1)$ is
\begin{align*}
g(x_1)\;& =\; -(\kappa_1x_1x_2^2 - \kappa_2x_2-\lambda \kappa_3x_1^3)\;|_{x_2=x_1+c_1}\\
&=\; -(\kappa_1-\lambda\kappa_3)x_1^3-2c_1\kappa_1x_1^2-(c^2_1\kappa_1-\kappa_2)x_1+c_1\kappa_2,
\end{align*}
where
$\lambda:=-\frac{\beta_{12}-\alpha_{12}}{\beta_{11}-\alpha_{11}}
=-\frac{\beta_{22}-\alpha_{22}}{\beta_{21}-\alpha_{21}}>0$.
Note that for any $\kappa_1>0$, $\kappa_2>0$ and for any $c_1\in {\mathbb R}$, $-2c_1\kappa_1$ and $c_1\kappa_2$ have different signs if $c_1\neq 0$. So the number of sign changes of the coefficients of $g(x_1)$ is at most $2$.
By Descartes' rule of signs (Theorem \ref{thm:des}),
$g(x_1)=0$ has at most $2$ positive roots. So, this network has at most $2$ positive steady states and by Theorem \ref{thm:multistable}, the network does not admit multistability.

For the network in Table \ref{tab:net2}-Row (\ref{row:1r1i8}),
the polynomial $g(x_1)$ is
$$g(x_1)\; =\; \kappa_1 - \kappa_2x_1x_2-\lambda \kappa_3x_1^3\;|_{x_2=x_1-c_1},$$
and the interval $I$ is $(\max\{0, c_1\}, +\infty)$, where
$\lambda:=-\frac{\beta_{12}-\alpha_{12}}{\beta_{11}-\alpha_{11}}
=-\frac{\beta_{22}-\alpha_{22}}{\beta_{21}-\alpha_{21}}<0$. 
So the number of sign changes of the coefficients of $g(x_1)$ is at most $2$. By Descartes' rule of signs (Theorem \ref{thm:des}),
$g(x_1)=0$ has at most $2$ positive roots. So,  by Theorem \ref{thm:multistable}, the network does not admit multistability. Similarly, we can show that
the network in Table \ref{tab:net2}-Row (\ref{row:1r1i9}) does not admit multistability. $\Box$


\newcounter{rowb}
\newcommand{\newrowb}{\refstepcounter{rowb}\arabic{rowb}}

\begin{table}[]
  \renewcommand\arraystretch{1.75}
  \tiny
  \centering
  \caption{All candidates for multistable networks in ${\mathcal G}_1$ with $3$ reactants and $2$ species}
  \label{tab:net2} (only networks in Rows \eqref{row:1r1i10}--\eqref{row:1r1i2} are multistable)
  \begin{tabular}{cc|c|r}
  \hline
  \cellcolor[HTML]{FFCE93}&\cellcolor[HTML]{FFCE93}  $\left(\alpha_{11},\alpha_{21},\beta_{11},\beta_{21},\alpha_{12},\alpha_{22}\right)$ &  \cellcolor[HTML]{FFCE93}Network  & \cellcolor[HTML]{FFCE93} $\beta_{12}$ and $\beta_{22}$ in ${\mathbb Z}_{\geq 0}$  \\ \hline



  \begin{tabular}[c]{@{}c@{}}
  	(\newrowb\label{row:1r1i4}) \\ \quad
  \end{tabular} & $(2,1,0,0,3,0)$ &
  \begin{tabular}[c]{@{}c@{}}
	$ 2X_1 +  X_2 \Leftrightarrow 0 $ \\
	$ 3X_1 \rightarrow \beta_{12}X_1 + \beta_{22}X_2 $
  \end{tabular}  &
  \begin{tabular}[r]{@{}r@{}}
	$\beta_{22}=\frac{1}{2}(\beta_{12}-3) $ \\
	$\beta_{12}-3>0$
  \end{tabular}

    \\ \hline

  \begin{tabular}[c]{@{}c@{}}
  	(\newrowb\label{row:1r1i5}) \\ \quad
  \end{tabular} & $(1,2,0,0,3,0)$ &
  \begin{tabular}[c]{@{}c@{}}
	$  X_1 + 2X_2 \Leftrightarrow 0 $ \\
	$ 3X_1 \rightarrow \beta_{12}X_1 + \beta_{22}X_2 $
  \end{tabular} &
  \begin{tabular}[r]{@{}r@{}}
	$\beta_{22}=2(\beta_{12}-3) $ \\
	$\beta_{12}-3>0$
  \end{tabular}

  \\ \hline

  \begin{tabular}[c]{@{}c@{}}
  	(\newrowb\label{row:1r1i6}) \\ \quad
  \end{tabular} & $(1,2,0,0,2,1)$ &
  \begin{tabular}[c]{@{}c@{}}
	$  X_1 + 2X_2 \Leftrightarrow 0 $ \\
	$ 2X_1 +  X_2 \rightarrow \beta_{12}X_1 + \beta_{22}X_2 $
  \end{tabular}  &
  \begin{tabular}[r]{@{}r@{}}
	$\beta_{22}=2(\beta_{12}-2)+1 $ \\
	$\beta_{12}-2>0$
  \end{tabular}
  \\ \hline

  \begin{tabular}[c]{@{}c@{}}
  	(\newrowb\label{row:1r1i7}) \\ \quad
  \end{tabular} & $(1,2,0,1,3,0)$ &
  \begin{tabular}[c]{@{}c@{}}
	$  X_1 + 2X_2 \Leftrightarrow X_2 $ \\
	$ 3X_1  \rightarrow \beta_{12}X_1 + \beta_{22}X_2 $
  \end{tabular}  &
  \begin{tabular}[r]{@{}r@{}}
	$\beta_{22}=\beta_{12}-3$ \\
	$\beta_{12}-3>0$
  \end{tabular}
  \\ \hline

  \begin{tabular}[c]{@{}c@{}}
  	(\newrowb\label{row:1r1i8}) \\ \quad
  \end{tabular} & $(0,0,1,1,3,0)$ &
  \begin{tabular}[c]{@{}c@{}}
	$  0 \Leftrightarrow  X_1 +  X_2 $ \\
	$ 3X_1  \rightarrow \beta_{12}X_1 + \beta_{22}X_2 $
  \end{tabular}  &
  \begin{tabular}[r]{@{}r@{}}
	$\beta_{22}=\beta_{12}-3 $ \\
	$\beta_{12}-3>0$
  \end{tabular}
   \\ \hline

  \begin{tabular}[c]{@{}c@{}}
  	(\newrowb\label{row:1r1i9}) \\ \quad
  \end{tabular} & $(0,0,1,1,2,1)$ &
  \begin{tabular}[c]{@{}c@{}}
	$  0 \Leftrightarrow  X_1 +  X_2 $ \\
	$ 2X_1 +  X_2  \rightarrow \beta_{12}X_1 + \beta_{22}X_2 $
  \end{tabular}  &
  \begin{tabular}[r]{@{}r@{}}
	$\beta_{22}=(\beta_{12}-2)+1 $ \\
	$\beta_{12}-2>0$
  \end{tabular}
   \\ \hline

\cellcolor[HTML]{6CDEFF}
  \begin{tabular}[c]{@{}c@{}}
  	(\newrowb\label{row:1r1i10}) \\ \quad
  \end{tabular} & $(0,1,1,0,2,1)$ &
  \begin{tabular}[c]{@{}c@{}}
	$  X_2 \Leftrightarrow X_1 $ \\
	$ 2X_1 +  X_2 \rightarrow \beta_{12}X_1 + \beta_{22}X_2 $
  \end{tabular}  &
  \begin{tabular}[r]{@{}r@{}}
	$\beta_{22}=-(\beta_{12}-2)+1 $ \\
	$\beta_{12}-2>0$
  \end{tabular}
  \\ \hline
\cellcolor[HTML]{6CDEFF}
  \begin{tabular}[c]{@{}c@{}}
  	(\newrowb\label{row:1r1i11}) \\ \quad
  \end{tabular} & $(1,0,0,2,2,1)$ &
  \begin{tabular}[c]{@{}c@{}}
	$  X_1 \Leftrightarrow 2X_2 $ \\
	$ 2X_1 +  X_2 \rightarrow \beta_{12}X_1 + \beta_{22}X_2 $
  \end{tabular} &
  \begin{tabular}[r]{@{}r@{}}
	$\beta_{22}=-2(\beta_{12}-2)+1 $ \\
	$\beta_{12}-2>0$
  \end{tabular}
  \\ \hline
\cellcolor[HTML]{6CDEFF}
\begin{tabular}[c]{@{}c@{}}
    (\newrowb\label{row:1r1i1}) \\ \quad
  \end{tabular}    & $(1,2,0,0,2,0)$&
  \begin{tabular}[c]{@{}c@{}}
	$  X_1 + 2X_2 \Leftrightarrow  0 $ \\
	$ 2X_1 \rightarrow \beta_{12}X_1 + \beta_{22}X_2 $
  \end{tabular}&
  \begin{tabular}[r]{@{}r@{}}
	$\beta_{22}=2(\beta_{12}-2) $ \\
	$\beta_{12}-2>0$
  \end{tabular}

  \\ \hline
\cellcolor[HTML]{6CDEFF}
  \begin{tabular}[c]{@{}c@{}}
  	(\newrowb\label{row:1r1i2}) \\ \quad
  \end{tabular} & $(1,2,0,1,2,0)$ &
  \begin{tabular}[c]{@{}c@{}}
	$  X_1 + 2X_2 \Leftrightarrow  X_2 $ \\
	$ 2X_1 \rightarrow \beta_{12}X_1 + \beta_{22}X_2 $
  \end{tabular} &
  \begin{tabular}[r]{@{}r@{}}
	$\beta_{22}=\beta_{12}-2 $ \\
	$\beta_{12}-2>0$
  \end{tabular}

  \\ \hline

  \end{tabular}


\end{table}

\begin{remark}\label{rmk:rrc}
In the proofs of Theorem \ref{thm:4-reactant} and Theorem \ref{thm:g1}, for each of the multistable networks, we found the witness by the software {\tt RealRootClassification} \cite{CDMMX2013} in {\tt Maple2020}.
The software can also outputs an open region (a semi-algebraic set) in the parameter space such that for any choice of parameters in the open region,
each of the multistable networks exhibits multistability.
\end{remark}


\section{Networks in ${\mathcal G}_2$: proof of Theorem \ref{thm:g2}} \label{sec:nig2}

\begin{definition}\cite[Definition 3.3]{Joshi:Shiu:Multistationary}
	Let $G$ be a reaction network that contains only one species $X_1$. Thus, each reaction of $G$ has the form $ aX_1 \rightarrow bX_1$, where $ a,b{\ge}0 $ and $ a {\neq} b$. Let $m$ be the number of (distinct) reactant complexes, and let $a_{1}<a_{2}<...<a_{m}$ be their stoichiometric coefficients. The \defword{arrow diagram} of $G$, denote $\rho = (\rho_1,...,\rho_m)$, is the element of $\{ \rightarrow,\leftarrow,\leftarrow\!\!\!\!\bullet\!\!\!\!\rightarrow \}^{m}$ defined by:
	\begin{equation}
		\rho_{i}=  \left\{
		\begin{array}{ll}
			\rightarrow \;\; & \text{if for all reactions $a_{i}X_1{\rightarrow}bX_1$ in $G$, it is the case that $b{>}a_i$} \\
			\leftarrow \;\; & \text{if for all reactions $a_{i}X_1{\rightarrow}bX_1$ in $G$, it is the case that $b{<}a_i$} \\
			\leftarrow\!\!\!\!\bullet\!\!\!\!\rightarrow & \text{otherwise.}
		\end{array}
		\right.
	\end{equation}
\end{definition}

\begin{definition}\cite[Definition 3.4]{Joshi:Shiu:Multistationary}
	For $T\in\mathbb{Z}_{{\ge}2}$, a \defword{$T$-alternating network} is a $1$-species network with exactly $T+1$ reactions and with arrow diagram $\rho\in\{ \rightarrow,\leftarrow \}^{T+1}$ such that $\rho_i=\rightarrow$ if and only if $\rho_{i+1}=\leftarrow$ for all $i\in\{ 1,...,T \}$.
\end{definition}

{\bf Proof of Theorem \ref{thm:g2}.}
	``$\Leftarrow$":
	For network $G$ with the form of the network \eqref{eq:unique}, $G$ has a 3-alternating subnetwork (i.e., $\{ 0{\rightarrow}X_1, 0{\leftarrow}X_1, 2X_1{\rightarrow}3X_1, 2X_1{\leftarrow}$ $3X_1 \}$) with arrow diagram $(\rightarrow,\leftarrow,\rightarrow,\leftarrow)$. By \cite[Theorem 3.6 2(c)]{Joshi:Shiu:Multistationary}, G admits at least $\lceil\frac{T}{2}\rceil = \lceil\frac{3}{2}\rceil = 2$ stable positive steady states. Thus, $G$ admits multistability.
	
	``$\Rightarrow$":
	For any network $G\in\mathcal{G}_2$, if $G$ has only one species, then $G$ has the form $ a_{0}X_1{\Leftrightarrow}a_{1}X_1,a_{2}X_1{\Leftrightarrow}a_{3}X_1$, and here, we have $a_{i}\in\{0,1,2,3\}$ since $G$ is at-most-3-reactant. If $G$ admits multistability (i.e., $cap_{stab}(G){\ge}2$), by Theorem \ref{thm:nonde} and Theorem \ref{thm:multistable}, $G$ admits at least 3 nondegenerate steady states. By \cite[Theorem 3.6 2(b)]{Joshi:Shiu:Multistationary}, $G$ has a 3-alternating subnetwork, so $G$ has 4 distinct reactant complexes with arrow diagram $(\rightarrow,\leftarrow,\rightarrow,\leftarrow)$. Thus, $G$ must have the form of the network \eqref{eq:unique}. $\Box$
	


\section{Discussion}\label{sec:dis}
Knowing the structures of small networks will help us to understand important networks in biology.
For instance, the multistability of the Huang-Ferrell mitogen-activated
protein kinase (MAPK) cascade with negative feedback can be inferred from a subnetwork with much fewer species and reactions \cite[Figure 1]{BP16}.
Following a study on a multistable network called  ERK (a model for dual-site phosphorylation and dephosphorylation of extracellular signal-regulated kinase) \cite[Figure 1]{OSTT2019}, a study on the maximum numbers of positive steady states for small networks has also been carried out by looking at the mixed volumes.

For future work, we propose the problems below.
\begin{itemize}
\item[(1)] Does there exist a network $G$ in ${\mathcal G}$ such that $cap_{pos}(G)=3$ but $cap_{stab}(G)<2$? We remark that for all small networks we have studied (see Table \ref{tab:net} and Table \ref{tab:net2}), if a network admits three positive steady states, then there are two stable ones.
\item[(2)]  Under which conditions does a network in ${\mathcal G}$ admit strictly more than $3$ positive steady states?
\item[(3)] For the set of networks ${\mathcal G}_i$ ($i\in \{0, 1, 2\}$), which subset is the smallest such that any network in this subset admits strictly more than $3$ positive steady states?
\end{itemize}
\section*{Acknowledgments}
The authors would thank Dr. Anne Shiu and Dr. Ang\'elica Torres for their suggestions on the first version of this paper.
The authors would also thank the anonymous referees. Their valuable comments  dramatically improved the presentation of this work.

\bigskip
\begin{center}
{\large\bf SUPPLEMENTARY MATERIAL}
\end{center}

Table \ref{tab:sup} lists all files at the online repository:
\url{https://github.com/HaoXUCode/MSRN-Supplement}


\begin{table}[h!]
\centering
\caption{Supporting Information Files}
\label{tab:sup}
{\scriptsize
\begin{tabular}{||l c l||}
 \hline
\cellcolor[HTML]{FFCE93} Name &\cellcolor[HTML]{FFCE93} File Type & \cellcolor[HTML]{FFCE93}Results \\ [0.5ex]
 \hline\hline
 1. \texttt{WitnessForMultistableNetworksInTheorem2.6.mw/.pdf} & \texttt{Maple/PDF} &  Theorem \ref{thm:g1}\\
 2. \texttt{EquivalenceClassesOfSmallReversibleNetworks4.mw/.pdf} & \texttt{Maple/PDF} & Lemma \ref{lm:g1} \\
 3. \texttt{EquivalenceClassesOfSmallReversibleNetworks3.mw/.pdf} & \texttt{Maple/PDF}& Lemma \ref{lm:g1} \\
 4. \texttt{EquivalenceClassesOfSmallReversibleNetworks2.mw/.pdf} & \texttt{Maple/PDF} & Lemma \ref{lm:g1} \\
 5. \texttt{EquivalenceClassesOfSmallReversibleNetworks1.mw/.pdf} & \texttt{Maple/PDF} & Lemma \ref{lm:g1} \\
6. \texttt{WitnessForMultistableNetworksInTheorem2.4.mw/.pdf} & \texttt{Maple/PDF} &  Theorem \ref{thm:4-reactant}\\
7. \texttt{EquivalenceClassesOfSmallIrreversibleNetworks.mw/.pdf} & \texttt{Maple/PDF} &  Lemma \ref{lm:4-reactant}\\
 \hline
\end{tabular}
}
\end{table}
\end{document}